\documentclass[11pt]{llncs}

\usepackage{amsmath}
\usepackage[utf8]{inputenc}
\usepackage[T1]{fontenc}
\usepackage{amsfonts}
\usepackage{amssymb}
\usepackage{float}
\usepackage{url}
\usepackage{shuffle}
 \usepackage[all]{xy}
 \usepackage{textgreek}
 \usepackage{enumitem}
 \usepackage{dsfont}
 \usepackage{stmaryrd}
\usepackage{pxfonts}
\usepackage{bbm}
\usepackage{tikz}
\usepackage[fancy]{tikz-inet}
\usetikzlibrary{automata}
\usetikzlibrary{arrows}
\usetikzlibrary{shapes}
\usetikzlibrary{decorations.pathmorphing}
\usetikzlibrary{fit}
\usepackage{ifthen}
\usepackage{stmaryrd} 

\tikzstyle{every picture}=[
  >=stealth', 
shorten >=1pt, 
node distance=1.44cm,
auto,
bend angle=45,
initial text=,
  every state/.style={inner sep=0.75mm, minimum size=1mm},
font=\small,
]
\def\cqfd{\hfill $\qed$}

\def\sc{\mathrm{sc}}

\newcommand{\IntEnt}[1]{\llbracket #1\rrbracket}
\newtheorem{conj}{Conjecture}

\title{A combinatorial approach for the state complexity of the Shuffle product}
\author{
    Pascal Caron  \and Jean-Gabriel Luque \and Bruno Patrou
    \thanks{\{Pascal.Caron,  Jean-Gabriel.Luque, Bruno.Patrou\}@univ-rouen.fr}
}
\institute{LITIS, Université de Rouen,\\ Avenue de l'Université,\\ 76801 Saint-\'Etienne du Rouvray Cedex,\\ France}
\begin{document}
\maketitle
\begin{abstract}
We investigate the state complexity of the shuffle operation on regular languages initiated by C\^ampeanu \textit{et al.} and  studied subsequently by Brzozowski \textit{et al.}. We shift  the problem into the combinatorics domain by turning  the problem of state accessibility into a problem of intersection of partitions. This allows us to develop new tools and to reformulate the conjecture of Brzozowski \textit{et al.} about the above-mentionned state complexity. 
\end{abstract}

\section{Introduction}
Studies on state complexity have been going on for more than forty years now. The seminal work of Maslov \cite{Mas70} which gives values (without proofs) for the state complexity of some operations: square root, cyclic shift and proportional removal, paves the way. From these foundations, a very active field of research was open mainly initiated by Yu et al \cite{YZS94}. 
Lots and lots of papers were produced and different sub-domains  appeared depending on whether the used automata are deterministic or not, whether the languages are finite or infinite, belongs to some classes (codes, star-free, $\ldots$) and so on. We focus here on the (complete) deterministic case for any language.

The state complexity of a rational language is the size of its minimal (complete deterministic) automaton and the state complexity of a rational operation is the maximal one of those languages obtained by applying this operation onto languages of fixed state complexities. 

The classical  approach is to compute an upper bound  and to provide a witness, that is a specific example reaching the bound which is then the desired state complexity.


In some cases, the classical method has to be enhanced by 
%
%
an algebraic approach consisting in 
 %
%
 building a witness for a certain class of rational operations by searching in a set of automata with as many  transition functions as possible. This method has the advantage of being applied to a large class of operations
and has been described independently by Caron \textit{et al.} in \cite{CHLP18}  as the monster approach and by Davies in \cite{Dav18} as the OLPA approach but was implicitly present in older papers like \cite{BJLRS16}, \cite{DO09}. 

The shuffle product of two languages is the set of words obtained by riffle shuffling any word of the first language  together with any word of the second one.
The shuffle product is a regular operation. While it is easy to describe in terms of automata \cite{Eil74}, its state complexity is notoriously difficult to establish \cite{BJLRS16,CSY02}.
In \cite{BJLRS16}, the authors use implicitly the notion of monsters which we explicit in this paper. 
In particular, Brzozowski \textit{et al.} introduced a class of tableaux allowing us to describe, in a combinatoric way, the states of the minimal DFA recognizing the shuffle product of two regular languages. By investigating the monoid of transformations through the point of view of modifiers and monsters, we give a more precise combinatorial  description of the underlying mechanism. Our main result consists in describing the state complexity as the cardinal of a class of combinatorial objects.

Proving the conjecture of Brzozowski \textit{et al.} is equivalent to prove that our class of objects is in bijection with the tableaux they consider. Although we do not achieve this goal, we provide numerous new tools and results in that context: related enumeration results,  generating functions, and partial description of the bijection.

The paper is organized as follows. Section \ref{sect-prel}  gives definitions and notations about automata and combinatorics. In Section  \ref{shuf-mons-tableau}, we recall the definition of the shuffle product and we drag it  into the realm of  \emph{monsters} and \emph{modifiers}. 

Section \ref{sect-comb-path} is devoted to the description of the paths of the shuffle automaton in a combinatoric way. Related enumeration formul\ae\ are studied in Section \ref{sect-enumeration}. In Section \ref{sect-sc}, we give an exact expression for the state complexity of the shuffle product.


\section{Preliminaries}\label{sect-prel}
Let $\Sigma$ denote a finite alphabet. A word $w$ over  $\Sigma$ is a finite sequence of symbols of $\Sigma$. 
  The set of all finite words over $\Sigma$ is denoted by $\Sigma ^*$.  The empty word is denoted by  $\varepsilon$. A language is a subset of $\Sigma^*$. The cardinality of a finite set $E$ is denoted by $\#E$,  the set of subsets of  $E$ is denoted by $2^E$ and the set of mappings of $E$ into itself is denoted by $E^E$. 

A  finite automaton (FA) is a $5$-tuple $A=(\Sigma,Q,I,F,\delta)$ where $\Sigma$ is the input alphabet, $Q$ is a finite set of states, $I\subset Q$ is the set of initial states, $F\subset Q$ is the set of final states and $\delta$ is the transition function from  $Q\times \Sigma$ to $2^Q$ extended in a natural way from $2^Q\times \Sigma^*$ to $2^Q$.
 A word $w\in \Sigma ^*$ is recognized by an FA $A$ if $\delta(I,w)\cap F\neq \emptyset$.
The language recognized by an FA $A$ is the set $L(A)$ of words recognized by $A$.
Two automata are said to be equivalent if they recognize the same language.
A state $q$ is accessible in an FA  if there exists a word $w\in \Sigma ^*$ such that $q\in \delta(I,w)$.
An FA is complete and deterministic (CDFA) if $\#I=1$ and for all $q\in Q$, for all $a\in \Sigma$, $\#\delta(q,a)= 1$. 
Let $D=(\Sigma,Q_D,i_D,F_D,\delta)$ be a CDFA. 
For any word $w$, we denote by $\delta^w$ the function $q\rightarrow\delta(q,w)$.
Two states $q_1,q_2$ of  $D$ are equivalent if for any word $w$ of $\Sigma^*$, $\delta(q_1, w)\in F_D$ if and only if $\delta(q_2, w)\in F_D$. Such an equivalence is denoted by $q_1\sim q_2$. A CDFA is  minimal if there does not exist any equivalent   CDFA  with less states and it is well known that for any CDFA, there exists a unique minimal equivalent one \cite{HU79}. Such a minimal CDFA  can be  obtained from $D$ by computing the accessible part of the automaton $D/_\sim=(\Sigma,Q_D/_\sim,[i_D],F_D/_\sim,\delta_{\sim})$ where for any $q\in Q_D$, $[q]$ is the $\sim$-class of the state $q$ and satisfies the property  $\delta_{\sim}([q],a)=[\delta(q,a)]$, for any $a\in \Sigma$. The number of its states is defined by $\#_{Min}(D)$.
In a minimal CDFA, any two distinct states are pairwise inequivalent. 
For any integer $n$, let us denote $\llbracket n\rrbracket$ for $\{0,\ldots, n-1\}$. 
%
 The state complexity of a regular language $L$ denoted by $\sc(L)$ is the number of states of its minimal CDFA. 
  Let ${\cal L}_n$ be the set of languages of state complexity $n$. The state complexity of a binary operation $\otimes$ is the  function $\sc_{\otimes}$ associating  $\max\{\sc(L_{1}\otimes L_{2})\mid L_1\in\mathcal{L}_{n_1},L_2\in \mathcal{L}_{n_2}\}$ with any couple of integers $n_1,n_2$.  
A witness for the binary operation $\otimes$ is a couple $(L_1, L_2)\in({\cal L}_{n_1}\times {\cal L}_{n_2})$ such that $\sc(L_1\otimes L_2)=\sc_{\otimes}(n_1,n_2)$. 

We  also need some background from finite transformation semigroup theory \cite{GM08}.
Let $n$ be an integer. A transformation $t$ is an element of $\IntEnt{n}^{\IntEnt{n}}$.
We denote by $it$ the image of $i$ under $t$. A transformation of $\IntEnt{n}$ can be represented by $t=[i_0, i_1, \ldots i_{n-1}]$ which means that $i_k=kt$ for each $k\in \IntEnt{n}$ and $i_k\in \IntEnt{n}$. A \textit{permutation} is a bijective transformation on $\IntEnt{n}$. 
A \textit{cycle} of length $\ell\leq n$  is a permutation $c$, denoted   by $(i_0,i_1,\ldots, i_{\ell-1})$, on a subset $I=\{i_0,\ldots ,i_{\ell-1}\}$ of $\IntEnt{n}$  where  $i_kc=i_{k+1}$ for $0\leq k<\ell-1$ and $i_{\ell-1}c=i_0$. A permutation is always a composition of disjoint cycles.


\section{Shuffle, tableaux and monsters.}\label{shuf-mons-tableau}
\subsection{The shuffle product}
The shuffle operation \cite{EM1953} on regular languages is classically implemented as follows \cite{Eil74}.
Let $K$ and $L$ be regular languages over an alphabet $\Sigma$ recognized by DFAs ${\cal K} = (Q_K,\Sigma,\delta_K,q_K,F_K)$ and ${\cal L} = (Q_L,\Sigma,\delta_L,q_L,F_L)$, respectively. Then $K\shuffle L$ is recognized by the NFA $N = (Q_K\times Q_L,\Sigma,\delta,(q_K,q_L),F_K\times F_L)$,
where $\delta((p,q),a) = {(\delta_K(p,a),q),(p,\delta_L(q,a))}$.

The state complexity of this operation was  first  studied by Campeanu, Salomaa and Yu \cite{CSY02}.
Later on, Brzozowski \textit{et al.}  \cite{BJLRS16} completed this study. 

Let $D=(2^{Q_K\times Q_L},\Sigma,\delta',\{(q_K,q_L)\},F')$ be the subset automaton of $N$. If
$|Q_K| = m$ and $|Q_L| = n$, then NFA $N$ has $mn$ states. It follows that DFA $D$ has
at most $2^{mn}$ reachable and pairwise distinguishable states.

Since states of $D$ belongs to $2^{Q_K\times Q_L}$, they are associated in a natural way to boolean tableaux of size $m\times n$, each cell of them being either empty or marked. Brzozowski \textit{et al.} prove that only tableaux with at least one marked cell on the first line and one marked cell on the first column can be reached. Such tableaux are called valid and their number, $$\mathbf f(m,n)=2^{mn-1}+2^{(m-1)(n-1)}(2^{(m-1)}-1)(2^{(n-1)}-1),$$ is an upper bound for the state complexity of the shuffle  operation. The authors also produce a couple of ternary languages $K$, $L$ for which all pairs of valid states 
are distinguishable.

The main difficulty is to prove that all valid states can be reached for some couple of languages $K, L$. This question of reachability only depends on the transition functions of $\cal{K}$ and $\cal{L}$. First, observe that the finality of states does not matter. Next, to reach any valid state in the most easier way, it is relevant to consider automata  ${\cal K}$ and ${\cal L}$ having a maximum of transitions. This is the idea of monsters detailed in \cite{CHLP18,Dav18} and formalized in the next section. Brzozowski \textit{et al.} implicitly use this notion to prove the result for any $n$ when $m\leq 5$. They also obtained the desired answer by computation when $m=n=6$, but they are unable to extend the result for any values of $m,n$. Even if we are not able to solve the conjecture, we provide a new approach for the question of the reachability for valid states, which especially allows to compute the exact value of the state complexity for the shuffle operation.

\subsection{Modifiers, monsters and state complexity}
The work of Brzozowski \textit{et al.} is implicitly based on the fact that the shuffle operation is a describable operation. Let us recall here the definition of a describable operation as described in \cite{CHLP18}.

\begin{definition}
A $2$-modifier $\mathfrak m$ is a  $4$-tuple of mappings $(\mathfrak Q,\mathfrak d,\iota,\mathfrak f)$ acting on $2$ CDFA $A_1,A_2$ with $A_j=(\Sigma, Q_j,i_j,F_j,\delta_ j), j\in \{1,2\}$ to build a CDFA $\mathfrak m(A_1,A_2)=(\Sigma,Q,i,F,\delta)$, where 
\[Q=\mathfrak Q(Q_1, i_1, F_1,Q_2,i_2,F_2)\;,\;\;
i=\iota (Q_1, i_1, F_1,Q_2,i_2,F_2)\;,\;\;
F=\mathfrak f (Q_1, i_1, F_1,Q_2,i_2,F_2) \text{ and }\]
$$\forall a\in \Sigma,\ \delta^a=\mathfrak d (Q_1, i_1,\delta^a_1, F_1,Q_2,i_2,\delta_2^a,F_2).$$
\end{definition}

\begin{definition}
We consider an operation $\otimes$ acting on a couple of languages defined on the same alphabet. The operation $\otimes$ is said to be \emph{describable}  if there exists a $2$-modifier $\mathfrak{m}$ 
 such that  for any couple of CDFA $(A_1,A_2)$, we have $L({\mathfrak m}(A_1,A_2))=L(A_1)\otimes L(A_2)$.  
\end{definition}

We are now able  to define the $\mathfrak{Shuf}$ modifier for the shuffle operation on automata. Only relevant parameters will appear in the definition.
$$\begin{array}{ll}
\mathfrak{Shuf}&=(\mathfrak Q,\mathfrak d,\iota,\mathfrak f)\text{ where}\\
&\mathfrak Q(Q_1,Q_2)=2^{Q_1\times Q_2}\\
&\mathfrak i(i_1,i_2)=\{(i_1,i_2)\}\\
&\mathfrak f(Q_1,F_1,Q_2,F_2)=\{E\in \frak Q(Q_1,Q_2)\mid E\cap (F_1\times F_2)\neq \emptyset\}\\
&\begin{array}{lcll}\mathfrak d(Q_1,\delta_1,Q_2,\delta_2):&2^{Q_1\times Q_2}&\rightarrow& 2^{Q_1\times Q_2}\\
&E&\rightarrow&\{(\delta_1(q_1),q_2)\mid (q_1,q_2)\in E\}\cup \{(q_1,\delta_2(q_2))\mid (q_1,q_2)\in E\}\end{array}
\end{array}$$

The classical construction for an automaton recognizing the language $L_1\shuffle L_2$ \cite{Eil74} for any pair $(L_1,L_2)$ of regular languages described, respectively, by two automata $A_1$ and $A_2$, is equivalent to the following statement:
$$L(\mathfrak{Shuf}(A_1,A_2))=L_1\shuffle L_2.$$
In other words, the shuffle is a describable operation.

A $2$-monster    is a couple of  DFAs of size $n_1,n_2$ having $n_1^{n_1}n_2^{n_2}$ letters representing  couple of  functions from $\IntEnt{n_1}$ to $\IntEnt{n_1}$ and from $\IntEnt{n_2}$ to $\IntEnt{n_2}$. There are $2^{n_1+n_2}$ different $2$-monsters  depending on the set of their final states. 
\begin{definition}
A $2$-monster  is a couple of automata  
$(M^{n_1}_{F_1},M^{n_2}_{F_2})$ where  $M^{n_j}_{F_j}=(\Sigma,\IntEnt{n_j},0, F_j,\delta_j)$ for $j\in \{1,2\}$ is defined by
\begin{itemize}
\item the common alphabet  $\Sigma=\IntEnt{n_1}^{\IntEnt{n_1}}\times \IntEnt{n_2}^{\IntEnt{n_2}}$, 
\item the set of states  $\IntEnt{n_j}$,
\item the initial state  $0$,
\item the set of final states  $F_j$,
\item the transition function $\delta_j$  defined for any $(a_1,a_2)\in \Sigma$ by $\delta_j(q,(a_1,a_2))={a_j}(q)$. 

\end{itemize}
\end{definition}
Notice that a symbol of the alphabet is assimilated to a single transition function from any state.

\subsection{Using monsters to compute state complexity} 
The idea behind the notion of monster is to define kind of universal pairs of automata maximizing the state complexity for any describable binary operation. It implies a common alphabet for these automata.

If an operation is describable, it is sufficient to study  the behavior of its modifiers over monsters to compute its state complexity.
From Theorem 1 in \cite{CHLP18} we obtain 
$$\mathrm{sc}_{\shuffle}(n_1,n_2)=\mathrm{max}\{\displaystyle\#_{Min}(\mathfrak{Shuf}(M^{n_1}_{F_1},M^{n_2}_{F_2}))\mid  F_j\subset \IntEnt{n_j}\}.$$

It means that a witness belongs to the set of monsters.\\


Brzozowski \textit{et al.} show the  following results that are translated in terms of modifier as :
\begin{itemize}
\item any accessible state in $\mathfrak{Shuf}(M^{n_1}_{F_1},M^{n_2}_{F_2})$ is valid,
\item any couple of valid states in $\mathfrak{Shuf}(M^{n_1}_{F_1},M^{n_2}_{F_2})$ can be distinguished by using successively many times the letters :
\begin{itemize}
\item $a=((0,\ldots,n_1-1),0)$
\item $b=(0,(0,\ldots,n_2-1))$
\item $c=((\alpha\rightarrow \delta_{0,\alpha},n_2-1)$ where $ \delta_{x,y}=\left\{\begin{array}{ll}1&\text{ if }x=y\\0&\text{ otherwise}\end{array}\right.$ denotes the Kronecker delta.
\end{itemize}
\end{itemize}


\section{The combinatorics of paths in the shuffle automaton}\label{sect-comb-path}
\subsection{ A first example}
We  illustrate the notions investigated in this section  with the following example.
Let $f_{a}$, $f_{b}$, and $f_{c}$ be three maps from $\IntEnt 4$ to itself such that $f_{b}(0)=0$, $f_{a}(0)=f_{c}(0)=2$, and $f_{b}(2)=f_{c}(2)=f_{c}(3)=3$.
Let $g_{a},g_{b}$ and $g_{c}$ be three maps from $\IntEnt 3$ to itself such that $g_{b}(1)=0$, $g_{a}(0)=g_{c}(2)=1$, and $g_{b}(0)=g_{c}(0)=g_{c}(1)=2$.  We denote by $x$ the pair $(f_{x},g_{x})$ for $x=a,b,c$.
These pairs of functions represent transition functions in a $2$-monster. These transition functions are drawn in Figure \ref{twoautomata} (missing transitions are not relevant for the example and are not drawn). 
	
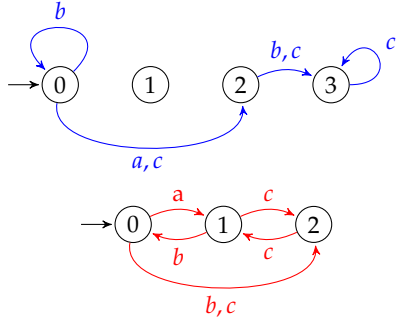
\begin{figure}[h]
	
\centerline{\begin{tikzpicture}[node distance=1.2cm, bend angle=25]
			\node[state,initial] (p0) {$0$};
			\node[state] (p1) [right of=p0] {$1$};
			\node[state] (p2) [right of=p1] {$2$};
			\node[state] (p3) [right of=p2] {$3$};
				\path[->]
        (p0) edge[bend left,out=-100, in=-85, looseness=.83,blue] node [swap]{$a,c$} (p2)
        (p0) edge [loop,blue] node [swap]{$b$} (p0)
        (p2) edge[bend left,blue] node {$b,c$} (p3)
        (p3) edge [out=0, in=60,loop,blue] node [swap]{$c$} (p3);
\end{tikzpicture}}

\centerline{\begin{tikzpicture}[node distance=1.2cm, bend angle=25]
			\node[state,initial] (p0) {$0$};
			\node[state] (p1) [right of=p0] {$1$};
			\node[state] (p2) [right of=p1] {$2$};
				\path[->]
        (p0) edge[bend left,red] node {a} (p1)
        (p0) edge[bend left,out=-100, in=-85, looseness=.83,red] node [swap]{$b,c$} (p2)
        (p1) edge[bend left,red] node {$b$} (p0)
        (p1) edge[bend left,red] node {$c$} (p2)
        (p2) edge[bend left,red] node {$c$} (p1);
\end{tikzpicture}}

\caption{Two automata figuring the transitions $a=(f_{a},g_{a}), b=(f_{b},g_{b})$, and $c=(f_{c},g_{c})$ in a $2$-monster\label{twoautomata}}
\end{figure}
We investigate the transitions in  $\mathfrak{Shuf}(M^{4}_{F_1},M^{3}_{F_2})$ drawn in figure \ref{TransShuf}.
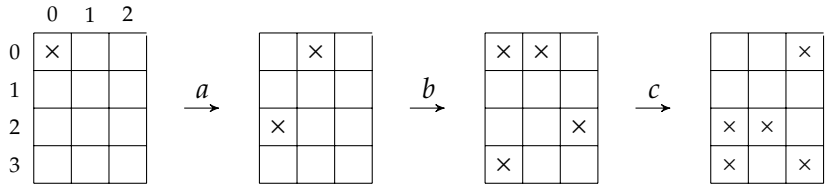
\begin{figure}[h]
	
 \centerline{             \begin{tikzpicture}[scale=0.5]
                    \draw[step=1.0,black, thin] (0,0) grid (3,4);
                    \node[scale=0.9] at (-0.5,3.5) {$0$};
                    \node[scale=0.9] at (-0.5,2.5) {$1$};
                    \node[scale=0.9] at (-0.5,1.5) {$2$};
                    \node[scale=0.9] at (-0.5,0.5) {$3$};
                    \node[scale=0.9] at (0.5,4.5) {$0$};
                    \node[scale=0.9] at (1.5,4.5) {$1$};
                    \node[scale=0.9] at (2.5,4.5) {$2$};
                    \node[scale=1] at (0.5,3.5) {$\times$};
                            \draw[->] (4,2) -- node[midway,above,scale=1.2] {$a$} (5,2) ;
                    \draw[step=1.0,black, thin] (6,0) grid (9,4);
                    \node[scale=1] at (7.5,3.5) {$\times$};
                    \node[scale=1] at (6.5,1.5) {$\times$};
                            \draw[->] (10,2) -- node[midway,above,scale=1.2] {$b$} (11,2) ;
                      \draw[step=1.0,black, thin] (12,0) grid (15,4);
                   \node[scale=1] at (12.5,3.5) {$\times$};
                   \node[scale=1] at (13.5,3.5) {$\times$};
                    \node[scale=1] at (14.5,1.5) {$\times$};
                   \node[scale=1] at (12.5,0.5) {$\times$};
                            \draw[->] (16,2) -- node[midway,above,scale=1.2] {$c$} (17,2) ;
                   \draw[step=1.0,black, thin] (18,0) grid (21,4);
                    \node[scale=0.85] at (20.5,3.5) {$\times$};
                   \node[scale=0.85] at (18.5,1.5) {$\times$};
                    \node[scale=0.85] at (18.5,0.5) {$\times$};
                   \node[scale=0.85] at (19.5,1.5) {$\times$};
                    \node[scale=0.85] at (20.5,0.5) {$\times$};
                  \end{tikzpicture}}
                  \caption{Some transitions in  $\mathfrak{Shuf}(M^{4}_{F_1},M^{3}_{F_2})$\label{TransShuf}}
\end{figure}
The tableaux represent subsets of $\IntEnt4\times\IntEnt3$; a symbol $\times$ is written in the cell $(i,j)$ if $(i,j)$ belongs to the subset. Following a transition symbol $(f,g)$, each tableau  is sent to another one constructed by merging the two tableaux obtained from the original one by  acting on the lines and  on the columns respectively by $f$ and $g$. An example is given in Figure \ref{Tableaux}.
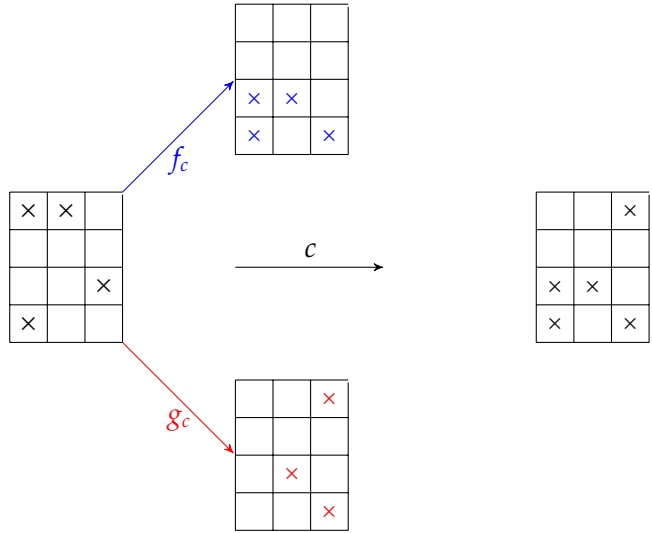
\begin{figure}[h]

 \centerline{             \begin{tikzpicture}[scale=0.5]
                      \draw[step=1.0,black, thin] (12,0) grid (15,4);
                   \node[scale=1] at (12.5,3.5) {$\times$};
                   \node[scale=1] at (13.5,3.5) {$\times$};
                    \node[scale=1] at (14.5,1.5) {$\times$};
                   \node[scale=1] at (12.5,0.5) {$\times$};
                            \draw[->,blue] (15,4) -- node[midway,below,scale=1.2] {$\color{blue} f_{c}$} (18,7) ;
                            \draw[->,red] (15,0) -- node[midway,below,scale=1.2] {$\color{red} g_{c}$} (18,-3) ;
                   \draw[step=1.0,black, thin] (18,5) grid (21,9);
                    \node[scale=0.85,red] at (20.5,-1.5) {$\times$};
                   \node[scale=0.85,red] at (19.5,-3.5) {$\times$};
                    \node[scale=0.85,red] at (20.5,-4.5) {$\times$}; -4.5
                            \draw[->] (18,2) -- node[midway,above,scale=1.2] {$c$} (22,2) ;
                   \draw[step=1.0,black, thin] (18,-5) grid (21,-1);
                   \node[scale=0.85,blue] at (18.5,6.5) {$\times$};
                    \node[scale=0.85,blue] at (18.5,5.5) {$\times$};
                   \node[scale=0.85,blue] at (19.5,6.5) {$\times$};
                    \node[scale=0.85,blue] at (20.5,5.5) {$\times$};
                   \draw[step=1.0,black, thin] (26,0) grid (29,4);
                    \node[scale=0.85] at (28.5,3.5) {$\times$};
                   \node[scale=0.85] at (26.5,1.5) {$\times$};
                    \node[scale=0.85] at (26.5,0.5) {$\times$};
                   \node[scale=0.85] at (27.5,1.5) {$\times$};
                    \node[scale=0.85] at (28.5,0.5) {$\times$};
                  \end{tikzpicture}}
                  \caption{Details of a transition\label{Tableaux}}
\end{figure}
In the aim to capture the path used to access to a state, we fill the cells of the tableaux by sets of integers as in figure \ref{TabPart}. More precisely, in a transition $T\displaystyle\mathop\rightarrow^{(f,g)} T'$ the tableau $T'$ is obtained from the tableau $T$ by moving the numbers according to the transformation (on lines) induced by $f$ and copying the numbers once shifted according to the transformation (on columns) induced by $g$.
\begin{figure}[h]
 \centerline{             \begin{tikzpicture}[scale=0.5]
                    \draw[step=1.0,black, thin] (0,0) grid (3,4);
                    \node[scale=0.9] at (-0.5,3.5) {$0$};
                    \node[scale=0.9] at (-0.5,2.5) {$1$};
                    \node[scale=0.9] at (-0.5,1.5) {$2$};
                    \node[scale=0.9] at (-0.5,0.5) {$3$};
                    \node[scale=0.9] at (0.5,4.5) {$0$};
                    \node[scale=0.9] at (1.5,4.5) {$1$};
                    \node[scale=0.9] at (2.5,4.5) {$2$};
                    \node[scale=1] at (0.5,3.5) {$1$};
                            \draw[->] (4,2) -- node[midway,above,scale=1.2] {$a$} (5,2) ;
                    \draw[step=1.0,black, thin] (6,0) grid (9,4);
                    \node[scale=1,red] at (7.5,3.5) {$2$};
                    \node[scale=1,blue] at (6.5,1.5) {$1$};
                            \draw[->] (10,2) -- node[midway,above,scale=1.2] {$b$} (11,2) ;
                      \draw[step=1.0,black, thin] (12,0) grid (15,4);
                   \node[scale=1,red] at (12.5,3.5) {$4$};
                   \node[scale=1,blue] at (13.5,3.5) {$2$};
                    \node[scale=1,red] at (14.5,1.5) {$3$};
                   \node[scale=1,blue] at (12.5,0.5) {$1$};
                            \draw[->] (16,2) -- node[midway,above,scale=1.2] {$c$} (17,2) ;
                   \draw[step=1.0,black, thin] (18,0) grid (21,4);
                    \node[scale=0.85,red] at (20.5,3.5) {$6\ 8$};
                   \node[scale=0.85,blue] at (18.5,1.5) {$4$};
                    \node[scale=0.85,blue] at (18.5,0.5) {$1$};
                   \node[scale=0.85] at (19.5,1.5) {$\color{blue}{2}\ \color{red}{7}$};
                    \node[scale=0.85] at (20.5,0.5) {$\color{blue}{3}\ \color{red}{5}$};
                  \end{tikzpicture}}
                  \caption{Tableaux of sets obtained by following a path.\label{TabPart}}
                  \end{figure}
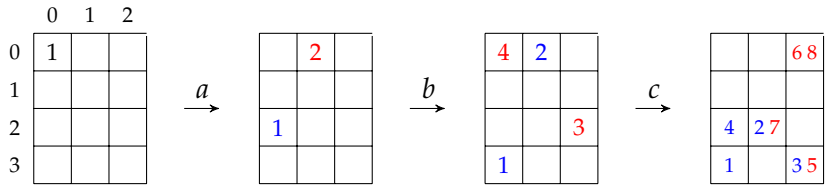
\bigskip                        
The non empty entries of a tableau obtained  this way partition the set $\{1,\dots,2^{k}\}$ for a given $k$. Such a configuration is completly described by a pair of vectors of disjoint parts $([\lambda_{1},\dots,\lambda_{m}],[\rho_{1},\dots,\rho_{n}])$ satisfying $\bigcup_{i}\lambda_{i}=\bigcup_{i}\rho_{i}=\{1,\dots,2^{k}\}$. The left vector is obtained by partitioning $\{1,\dots,2^{k}\}$ with respect to the lines of the tableau and the right vector is obtained by partitioning it with respect to the columns. Conversely, knowing the pair  $([\lambda_{1},\dots,\lambda_{m}],[\rho_{1},\dots,\rho_{n}])$, we easily reconstitute the tableau $T$ by noting that $T[i,j]=\lambda_{i}\cap\rho_{j}$; see Figure \ref{Tab2Vect} for an example.
\begin{figure}[h]
 \centerline{             \begin{tikzpicture}[scale=0.5]
                                     \draw[step=1.0,black, thin] (18,0) grid (21,4);
                    \node[scale=0.85,black] at (20.5,3.5) {$6\ 8$};
                   \node[scale=0.85,black] at (18.5,1.5) {$4$};
                    \node[scale=0.85,black] at (18.5,0.5) {$1$};
                   \node[scale=0.85] at (19.5,1.5) {$2\ 7$};
                    \node[scale=0.85] at (20.5,0.5) {$3\ 5$};
                    \node[scale=0.85,black] at (20.5,-0.5) {$\{3,5,6,8\}$};
                    \node[scale=0.85,black] at (19.5,4.5) {$\{2,7\}$};
                    \node[scale=0.85,black] at (18.5,-0.5) {$\{1,4\}$};
                    \node[scale=0.85,black] at (22,0.5){$\{1,3,5\}$};
                    \node[scale=0.85,black] at (22,1.5){$\{2,4,7\}$};
                    \node[scale=0.85,black] at (22,2.5){$\emptyset$};
                    \node[scale=0.85,black] at (22,3.5){$\{6,8\}$};
                  \node[scale=1.5,black] at (23.5,2) {$\sim$};
                  \node[scale=0.85,black] at (30,2) {$({\color{black}[\{6,8\},\emptyset,\{2,4,7\},\{1,3,5\}]},{\color{black}[\{1,4\},\{2,7\},\{3,5,6,8\}]})$};\end{tikzpicture}
                  }
                  \caption{A tableau and its associated pair of vectors of parts.\label{Tab2Vect}}
                  \end{figure}
Hence, the process allowing to obtain a tableau from a path is easily translated in terms of pairs of vectors. For instance, in Figure \ref{Path2Vect}, we follows the same path as in Figure \ref{TabPart} but we construct the associated pairs of vectors instead of the tableaux.
 \begin{figure}[h]
	 \[{}
	 \begin{array}{c}
		 ([\{1\},\emptyset,\emptyset,\emptyset],[\{1\},\emptyset,\emptyset])\\\downarrow a\\
		 ([\{{\color{red}2}\},\emptyset,\{{\color{blue}1}\},\emptyset],[\{{\color{blue}1}\},\{{\color{red}2}\},\emptyset])\\\downarrow b\\
		 ([\{{\color{blue}2},{\color{red}4}\},\emptyset,\{{\color{red}3}\},\{{\color{blue}1}\}],[\{{\color{blue}1},{\color{red}4}\},
		 \{{\color{blue}2}\},\{{\color{red}3}\}])\\\downarrow c\\
		 (\Lambda,P)=([\{{\color{red}{6,8}}\},\emptyset,\{{{\color{blue}2,4,\color{red}7}}\},\{{\color{blue}{1,3}},{\color{red}5}\}],[\{{\color{blue}{1,4}}\},
		 \{{\color{blue}2},{\color{red}7}\},\{{\color{blue}3},{\color{red}5,6,8}\}])
	 \end{array}
	 \]
	 \caption{A pair of vectors obtained from a path.\label{Path2Vect}}
 \end{figure}
It is easy to see that only some pairs can be constructed in such a way. For instance, the construction implies that $1$ always belongs to the first entry of the right vector and $2^{k}$ in the first entry of the left vector. But in the aim to completely describe the valid pairs in a combinatorics way, we need more rules. Indeed, observe the bottom pair $(\Lambda,P)$ in Figure \ref{Path2Vect}. The fact that $1$ and $4$ belongs to a same set in $P$ implies that $5$ and $8$ also belongs to a same set in $P$. Indeed, the set $\{5,8\}$ is the (shifted) image of $\{1,4\}$ which was obtained in the previous step.  In the same way, the fact that $6$ and $8$ belongs to a same set in $\Lambda$ implies that $2$ and $4$ also belongs to a same set in $\Lambda$. These properties are captured by the notion of $U$-pair formally described in the next section.
In what follows, we investigate the combinatorics of such objects and their relation with the state complexity of the shuffle product.
\subsection{A combinatorial representation of the paths}
 We define the \emph{left domain} (resp. \emph{right domain}) of  $E\subset \IntEnt{m}\times\IntEnt{n}$ as $\mathcal D_{L}(E)=\{i\mid (i,j)\in E\mbox{ for some }j\}$ (resp. $\mathcal D_{R}(E)=\{j\mid (i,j)\in E\mbox{ for some }i\}$). 
We associate with any transition $\mathfrak t=(E_{1},(f,g),E_{2})$ in the automaton $\mathfrak {Shuf}(M^{n_1}_{F_1},M^{n_2}_{F_2})$ 
its \emph{useful step} defined as the triplet $\mathfrak t^{u}=(E_{1},(f|_{\mathcal D_{L}(E_1)},g|_{\mathcal D_{R}(E_1)}),E_{2})$ 
where the notation $f|_{D}$ means the restriction of the map $f$ to the subdomain $D$. 
Consider a path $\mathfrak p=(\mathfrak t_{1},\dots,\mathfrak t_{k})$, its associated
 \emph{useful path} is defined by $\mathfrak p^{u}=(\mathfrak t_{1}^{u},\dots,\mathfrak t_{k}^{u})$.
 \begin{example}
	In Figure \ref{TransShuf}, we consider the transition labeled by $c$  given 
	by the triplet $\mathfrak t =(\{(0,0),(0,1),(2,2),(3,0)\},(f_{c},g_{c}),\{(0,2),(2,0),(2,1),(3,0),(3,2)\})$.\\ 
	Since $\mathcal D_{L}(\{(0,0),(0,1),(2,2),(3,0)\})=\{0,2,3\}$ and $\mathcal D_{R}(\{(0,0),(0,1),(2,2),(3,0)\})=\{0,1,2\}$, we have
	$\mathfrak t=(\{(0,0),(0,1),(2,2),(3,0)\},(f_{c}|_{\{0,2,3\}},g_{c}|_{\{0,1,2\}}),\{(0,2),(2,0),(2,1),(3,0),(3,2)\}).$ The useful path associated to the path drawn in Figure \ref{TransShuf} is
	\[{}\begin{array}{l}
	\left((\{(0,0)\},(f_{a}|_{\{0\}},g_{a}|_{\{0\}}),\{(0,1),(2,0)\}),\right. \\
	(\{(0,1),(2,0)\},(f_{a}|_{\{0,2\}},g_{a}|_{\{0,1\}}),\{(0,1),(0,2),(2,2),(3,0)\}),{}\\
	\left.(\{(0,0),(0,1),(2,2),(3,0)\},(f_{c}|_{\{0,2,3\}},g_{c}|_{\{0,1,2\}}),\{(0,2),(2,0),(2,1),(3,0),(3,2)\})\right).\end{array}
	\]
\end{example}
 
 For any accessible state $E$ in  $\mathfrak {Shuf}(M^{n_1}_{F_1},M^{n_2}_{F_2})$  and any path $\mathfrak p=(\mathfrak t_{1},\dots,\mathfrak t_{k})$ from $\{(0,0)\}$ to $E$, we define recursively a pair of vectors of sets $\mathcal P(\mathfrak p)=([\lambda_{0},\dots,\lambda_{m-1}],[\rho_{0},\dots,\rho_{n-1}])$ 
 \begin{itemize}
	 \item If $k=0$ (this means that $\mathfrak p=()$ and $E=\{(0,0)\}$) then $\mathcal P(\mathfrak p)=([\{1\},\emptyset,\dots,\emptyset],[\{1\},\emptyset,\dots,\emptyset])$,
	 \item If $k>0$ then denotes 
	  $\mathcal P(\mathfrak t_{1},\dots,\mathfrak t_{k-1})=(\Lambda,P)$ and $\mathfrak t_{k}=(E_{k-1},(f,g),E_{k})$. We set
	  \[{}
	   \mathcal P(\mathfrak p)=[(\Lambda\cdot {f}) \cup \Lambda^{\uparrow},P\cup (P\cdot{g})^{\uparrow}]
	   \]
	   with the notation

	   $$[\pi_{0},\dots,\pi_{\ell-1}]\cdot{h}=[\bigcup_{q\cdot h=0}\pi_{q},\dots,\bigcup_{q\cdot h=\ell-1}\pi_{q}],$$ 
	   	   $$[\pi_{0},\dots,\pi_{\ell-1}]\cup [\pi'_{0},\dots,\pi'_{\ell-1}]=[\pi_{0}\cup\pi'_{0},\dots,{}
	   \pi_{\ell-1}\cup\pi'_{\ell-1}],$$ and
	   	   $$[\pi_{0},\dots,\pi_{\ell-1}]^{\uparrow}=[\{q+r\mid q\in\pi_{0}\},\dots,
	   \{q+r\mid q\in\pi_{\ell-1}\}],$$
where $r=\max\bigcup_{q}\pi_{q}$.
 \end{itemize}
 Such an object is called a \emph{U-pair} and for a given $(m,n)$ the set of the U-pairs is denoted by $\mathcal U_{m,n}$.
 \begin{example}
	 Consider again the path  $\mathfrak p=(\mathfrak t_{a},\mathfrak t_{b},\mathfrak t_{c})$ drawn in Figure \ref{TransShuf}. We compute successively 
		 $\begin{array}{lll}\mathcal P(())&=&([\{1\},\emptyset,\emptyset,\emptyset],[\{1\},\emptyset,\emptyset]),\\
		 \mathcal P((\mathfrak t_{a}))&=&\left( [\emptyset,\emptyset,\{1\},\emptyset]\cup[\{1\},\emptyset,\emptyset,\emptyset]^{\uparrow},[\{1\},\emptyset,\emptyset]
		 \cup [\emptyset,\{1\},\emptyset]^\uparrow \right)=
		 ([\{2\},\emptyset,\{1\},\emptyset],[\{1\},\{2\},\emptyset]),\\
		 \mathcal P((\mathfrak t_{a},\mathfrak t_{b}))&=
		& \left([\{2\},\emptyset,\emptyset,\{1\}]\cup [\{2\},\emptyset,\{1\},\emptyset]^\uparrow,[\{1\},\{2\},\emptyset]\cup[\{2\},\emptyset,\{1\}]^\uparrow\right)\\
		 &=&([\{2,4\},\emptyset,\{3\},\{1\}],[\{1,4\},\{2\},\{3\}]),\\
		\mathcal P(\mathfrak p)&=&\left([\emptyset,\emptyset,\{2,4\},\{1,3\}]\cup[\{2,4\},\emptyset,\{3\},\{1\}]^\uparrow,
		 [\{1,4\},\{2\},\{3\}]\cup[\emptyset,\{3\},\{1,2,4\}]^\uparrow\right)\\
		 &=&([\{6,8\},\emptyset,\{2,4,7\},\{1,3,5\}],[\{1,4\},\{2,7\},\{3,5,6,8\}]).\end{array}$
		 
	 This is exactly the process described in Figure \ref{Path2Vect}.
 \end{example}
 The following proposition compiles some basic facts about U-pairs.
 \begin{proposition}\label{basicfacts}Let $\mathcal P((\mathfrak t_{1},\dots,\mathfrak t_{k}))=([\lambda_{0},\dots,\lambda_{m-1}],
 [\rho_{0},\dots,\rho_{n-1}])$. We have
	 \begin{enumerate}
		 \item $\bigcup_{q}\lambda_{q}=\bigcup_{q}\rho_{q}=\{1,\dots,2^{k}\}$,
		 \item $\lambda_{i}\cap\lambda_{j}=\emptyset$ for all  $i\neq j$,
		 \item $\rho_{i}\cap\rho_{j}=\emptyset$ for all  $i\neq j$,
		 \item $\mathcal P((\mathfrak t_{1},\dots,\mathfrak t_{k}))= \mathcal P((\mathfrak t_{1}^{u},\dots,\mathfrak t_{k}^{u}))$.
	 \end{enumerate}
 \end{proposition}
We notice that $\mathcal U_{m,n}$ is a graded set $\mathcal U_{m,n}=\bigcup_{k\geq 0}\mathcal U_{m,n}^{(k)}$ where
 \begin{equation}\mathcal U_{m,n}^{(k)}=\left\{\left([\lambda_{0},\dots,\lambda_{m-1}],[\rho_{0},\dots,\rho_{n-1}]\right)\in \mathcal U_{m,n}\mid\bigcup_{q} \lambda_{q}=\bigcup_{q}\rho_{q}=\{1,\dots,2^{k}\}\right\},\end{equation}
 that is the set of images by $\mathcal P$ of  paths of length $k$ and source $\{(0,0)\}$.
 \begin{example}\label{ExU22}\ \\
\medskip
	 $\begin{array}{ll}
	 \mathcal U_{2,2}^{(0)}=&\{([\{1\},\emptyset],[\{1\},\emptyset])\},\\
	 \mathcal U_{2,2}^{(1)}=&\{([\{1,2\},\emptyset],[\{1,2\},\emptyset]),{}
	 ([\{1,2\},\emptyset],[\{1\},\{2\}]), ([\{2\},\{1\}],[\{1,2\},\emptyset]), ([\{2\},\{1\}],[\{1\},\{2\}])\},\\
	 \mathcal U_{2,2}^{(2)}=&\{([\{1,2,3,4\},\emptyset],[\{1,2,3,4\},\emptyset]),{}
	 ([\{1,2,3,4\},\emptyset],[\{1,2\},\{3,4\}]),
	 ([\{3,4\},\{1,2\}],[\{1,2,3,4\},\emptyset]),\\ 
	 &([\{3,4\},\{1,2\}],[\{1,2\},\{3,4\}]),
	 ([\{1,2,3,4\},\emptyset],[\{1,3,4\},\{2\}]), 
	 ([\{1,2,3,4\},\emptyset],[\{1,3\},\{2,4\}]),\\
	 &([\{1,2,3,4\},\emptyset],[\{1,4\},\{2,3\}]), ([\{1,2,3,4\},\emptyset],[\{1\},\{2,3,4\}]),{}
	 ([\{3,4\},\{1,2\}],[\{1,3,4\},\{2\}]), \\
	 &([\{3,4\},\{1,2\}],[\{1,3\},\{2,4\}]),
	 ([\{3,4\},\{1,2\}],[\{1,4\},\{2,3\}]), ([\{3,4\},\{1,2\}],[\{1\},\{2,3,4\}]),
	\\
	 &([\{1,2,4\},\{3\}],[\{1,2,3,4\},\emptyset]),([\{1,2,4\},\{3\}],[\{1,2\},\{3,4\}]),
	 ([\{1,4\},\{2,3\}],[\{1,2,3,4\},\emptyset]),\\
	 &([\{1,4\},\{2,3\}],[\{1,2\},\{3,4\}]),
	 ([\{2,4\},\{1,3\}],[\{1,2,3,4\},\emptyset]),([\{2,4\},\{1,3\}],[\{1,2\},\{3,4\}]),\\
	 &([\{4\},\{1,2,3\}],[\{1,2,3,4\},\emptyset]),([\{4\},\{1,2,3\}],[\{1,2\},\{3,4\}]),
	 ([\{1,2,4\},\{3\}]),[\{1,3,4\},\{2\}],\\
	 &([\{1,2,4\},\{3\}],[\{1,3\},\{2,4\}]),
	 ([\{1,2,4\},\{3\}],[\{1,4\},\{2,3\}]),([\{1,2,4\},\{3\}],[\{1\},\{2,3,4\}]),\\
	 &([\{1,4\},\{2,3\}],[\{1,3,4\},\{2\}]),([\{1,4\},\{2,3\}],[\{1,3\},\{2,4\}]),
	 ([\{1,4\},\{2,3\}],[\{1,4\},\{2,3\}),\\
	 &([\{1,2,4\},\{2,3\}],[\{1\},\{2,3,4\}]),
	 ([\{2,4\},\{1,3\}],[\{1,3,4\},\{2\}]),([\{2,4\},\{1,3\}],[\{1,3\},\{2,4\}]),\\
	 &([\{1,4\},\{1,3\}],[\{1,4\},\{2,3\}),([\{2,4\},\{1,3\}],[\{1\},\{2,3,4\}]),
	 ([\{4\},\{1,2,3\}],[\{1,3,4\},\{2\}]),\\
	 &([\{4\},\{1,2,3\}],[\{1,3\},\{2,4\}]),
	 ([\{4\},\{1,2,3\}],[\{1,4\},\{2,3\}]),([\{4\},\{1,2,3\}],[\{1\},\{2,3,4\}])\}.
	 \end{array}
	 $
 \end{example}
 \begin{proposition}
	 Each graded component $\mathcal U_{m,n}^{(k)}$ is in  a one to one correspondence with the set  of  useful paths  of length $k$ and source $\{(0,0)\}$.
 \end{proposition}
 \begin{proof}
	 From the last point of proposition \ref{basicfacts} it suffices to prove that if $\mathfrak p^{u}\neq \mathfrak p'^{u}$ then 
	 $ \mathcal P(\mathfrak p)\neq  \mathcal P(\mathfrak p')$. We show this statement by induction on $k$. 
	 For $k=0$ the result is obvious. 
	 Suppose $\mathfrak p=(\mathfrak t_{1},\dots,\mathfrak t_{k})$ and 
	 $\mathfrak p'=(\mathfrak t'_{1},\dots,\mathfrak t'_{k})$ are two paths of length $k>0$ such that
	  $\mathfrak p^{u}\neq \mathfrak p'^{u}$. 
	  Let us denote $(\Lambda, P)=\mathcal P((\mathfrak t_{1}^{u},\dots,\mathfrak t_{k-1}^{u}))$ and $(\Lambda', P')=\mathcal P(({\mathfrak t'}_{1}^{u},\dots,{\mathfrak t'}_{k-1}^{u}))$.
	  If $(\mathfrak t_{1}^{u},\dots,\mathfrak t_{k-1}^{u})\neq 
	  ({\mathfrak t'}_{1}^{u},\dots,{\mathfrak t'}_{k-1}^{u})$ then by induction 
	 $(\Lambda, P)\neq  (\Lambda', P')$. For any $f, f'\in\IntEnt m$ and $g, g'\in\IntEnt n$, we have
	 $((\Lambda\cdot{f})\cup \Lambda^\uparrow, P\cup (P\cdot{g})^\uparrow)\neq ((\Lambda'\cdot {f'})\cup \Lambda'^\uparrow, P'\cup  (P'\cdot{g'})^\uparrow)$
	 because $\Lambda\neq \Lambda'$ or $P\neq P'$ and thus $\mathcal P(\mathfrak p)\neq  \mathcal P(\mathfrak p')$.
	 If $(\mathfrak t_{1}^{u},\dots,\mathfrak t_{k-1}^{u})=({\mathfrak t'}_{1}^{u},\dots,{\mathfrak t'}_{k-1}^{u})$ then $\mathfrak t_{k}^{u}\neq{\mathfrak t'}_{k}^{u}$.	 
	  There exist $(f,g)\neq (f',g')$ such that $\mathcal P(\mathfrak p)=((\Lambda\cdot {f})\cup \Lambda^\uparrow, P\cup (P\cdot{g})^\uparrow)$ and  $\mathcal P(\mathfrak p')=((\Lambda\cdot {f'})\cup \Lambda^\uparrow, P\cup (P\cdot{g'})^\uparrow)$ and  there exists $i$ with $\lambda_{i}\neq\emptyset$ and $f(i)\neq f'(i)$ or there exists $j$ with $\rho_{j}\neq\emptyset$ and $g(j)\neq g'(j)$. Hence $\Lambda\cdot f\neq \Lambda\cdot f'$ or $P\cdot g\neq P\cdot g'$.\cqfd
 \end{proof}
 A $k$-\emph{EXPcomposition} of size $n$ is a vector of sets $[\pi_{1},\dots,\pi_{n}]$ satisfying $\bigcup_{q}\pi_{q}=\{1,\dots,2^{k}\}$ and $\pi_{q}\cap \pi_{q'}\neq\emptyset$ implies $q=q'$.
 A $k$-\emph{Rvalid} vector of size $n$ is a $k$-\emph{EXPcomposition} $[\rho_{1},\dots,\rho_{n}]$ such that $1\in\rho_{1}$ and  for any $k'<k$ and any $i,j\leq 2^{k'}$, if $i, j\in \rho_{\alpha}$ for some $\alpha$ then $i+2^{k'}, j+2^{k'}\in \rho_{\beta}$ for some $\beta$. The set of $k$-Rvalid vector of size $n$ is denoted by $\mathcal R_{n}^{(k)}$. We define also  a $k$-\emph{Lvalid} vector of size $n$ as a $k$-\emph{EXPcomposition} $[\lambda_{1},\dots,\lambda_{m}]$ such that $2^{k}\in\lambda_{1}$ and  for any $k'<k$ and any $i,j\leq k'$, if $2^{k}+1-i, 2^{k}+1-j\in \lambda_{\alpha}$ for some $\alpha$ then $2^{k}+1-(i+2^{k'}), 2^{k}+1-(j+2^{k'})\in \lambda_{\beta}$ for some $\beta$. The set of $k$-Lvalid vector of size $m$ is denoted by $\mathcal L^{(k)}_{m}$. Remark that we can obtain any element of  $\mathcal L^{(k)}_{m}$ from an element of $\mathcal R_{n}^{(k)}$ by replacing the only occurrence of $i$ by $2^{k}+1-i$ for any $1\leq i\leq 2^k$.
 \begin{lemma} For any $k>0$,
	 we have $\mathcal L^{(k)}_{m}=\{(\Lambda\cdot {f})\cup \Lambda^\uparrow\mid \Lambda\in \mathcal L^{(k-1)}_{m}, f\in\IntEnt m^{\IntEnt m}\}$ and 
	 $\mathcal R^{(k)}_{n}=\{P\cup (P\cdot {g})^\uparrow\mid P\in \mathcal R^{(k-1)}_{n}, g\in\IntEnt n^{\IntEnt n}\}$.
 \end{lemma}
 \begin{proof}
	 We remark that the elements of $\mathcal R^{(k)}_{\ell}$ are in a one to one correspondence with the elements of $\mathcal L^{(k)}_{\ell}$. Indeed, a bijection  sends any element of   $\mathcal R_{\ell}^{(k)}$  to an element $\mathcal L^{(k)}_{\ell}$ by replacing the only occurrence of $i$ by $2^{k}+1-i$ for any $1\leq i\leq 2^k$. So it suffices to prove the result for  $\mathcal R^{(k)}_{n}$.\\ 
	 Let $P\in \mathcal R^{(k)}_{n}$ and consider $P'=[\rho'_{1},\dots,\rho'_{n}]$ obtained from $P$ by erasing all the numbers strictly greater than $2^{k-1}$. We check that $P'\in\mathcal R^{(k-1)}_{n}$ and we set $P''=[\rho''_{1},\dots,\rho''_{n}]$ such that $P=P'\cup P''^\uparrow$. Since $P\in\mathcal R^{(k)}_{n}$, for any $1\leq i\leq m$ one has $\rho'_{i}\subset \rho''_{\alpha_{i}}$ for some integer $\alpha_{i}$. Setting $i\cdot g=\alpha_{i}$, one obtains $P'\cdot{g}=P''$ and so $\mathcal R^{(k)}_{n}\subset\{P\cup (P\cdot{g})^\uparrow\mid P\in \mathcal R^{(k-1)}_{n}, g\in\IntEnt n^{\IntEnt n}\}$.
	 
	 Conversely, consider $P'=P\cup (P\cdot{f})^\uparrow=[\rho'_{1},\dots,\rho'_{n}]$ with $P=[\rho_{1},\dots,\rho_{n}]\in\mathcal R_{n}^{(k-1)}$ and $g\in\IntEnt n^{\IntEnt n}$. Straightforwardly, one obtains $\bigcup_{q}\rho'_{q}=\{1,\dots,2^{k}\}$, $1\in\rho'_{1}$, and $\rho'_{q}\cap\rho'_{q'}\neq\emptyset$ implies $q=q'$. Furthermore, since $P\in\mathcal R_{n}^{(k-1)}$, one has for any $k'<k-1$ and any $i,j\leq 2^{k'}$, if $i, j\in \rho'_{\alpha}$ for some $\alpha$ then $i+2^{k'}, j+2^{k'}\in \rho'_{\beta}$ for some $\beta$. Let $i,j\leq 2^{k-1}$ such that $i,j\in\lambda_{\alpha}$ for some $\alpha$ one
	  has $i+2^{k-1},j+2^{k-1}\in \rho'_{\alpha\cdot g}$ and so we deduce that $P'\in\mathcal R_{n}^{(k)}$ and
	   $\{P\cup (P\cdot{g})^\uparrow\mid P\in \mathcal R^{(k-1)}_{n}, g\in\IntEnt n^{\IntEnt n}\}\subset \mathcal R^{(k)}_{n}$.\cqfd
 \end{proof}
 As a direct consequence, one obtains the following result.
 \begin{proposition}
	  Each graded component $\mathcal U_{m,n}^{(k)}$ splits into the cartesian product
	 \begin{equation}
		 \mathcal U_{m,n}^{(k)}=\mathcal L_{m}^{(k)}\times \mathcal R_{n}^{(k)}.
	 \end{equation}
 \end{proposition}
 \begin{proof}
	 By induction on $k$ from the definition of  $\mathcal U_{m,n}^{(k)}$ and the previous lemma.\cqfd
 \end{proof}
 \begin{example}
	 We have
	  $$
	 \mathcal L^{(2)}_{2}=\{[\{1,2,3,4\},\emptyset],[\{3,4\},\{1,2\}],[\{1,2,4\},\{3\}],
	 [\{1,4\},\{2,3\}],[\{2,4\},\{1,3\}],[\{4\},\{1,2,3\}]\}
	 $$
	 and
	 $$\mathcal R^{(2)}_{2}=\{[\{1,2,3,4\},\emptyset],[\{1,2\},\{3,4\}],[\{1,3,4\},\{2\}],
	 [\{1,4\},\{2,3\}],[\{1,3\},\{2,4\}],[\{1\},\{2,3,4\}]\}$$
	 Compare $\mathcal L^{(2)}_{2}\times  \mathcal R^{(2)}_{2}$ to $\mathcal U^{(2)}_{2,2}$ as described in Example \ref{ExU22}.
 \end{example}
 \begin{example}
	 In Figure \ref{Path2Vect}, the pair obtained by the last transition satisfies
	 $[\{6,8\},\emptyset,\{2,4,7\},\{1,3,5\}]=([\{2,4\},\emptyset,\{3\},\{1\}]\cdot {f_{c}})\cup [\{2,4\},\emptyset,\{3\},\{1\}]^\uparrow$
	 and $[\{1,4\},\{2,7\},\{3,5,6,8\}]=[\{1,4\},\{2\},\{3\}]\cup ([\{1,4\},\{2\},\{3\}]\cdot{g_{c}})^\uparrow$, where
	 $f_{c}(0)=2$, $f_{c}(2)=f_{c}(3)=3$, $g_{c}(0)=g_{c}(1)=2$, and $g_{c}(2)=1$.
 \end{example}

 \section{Enumeration formul\ae\ for U-pairs}\label{sect-enumeration}
 \subsection{Counting the successors of U-pairs}
 Let $\mathcal L_{m,\ell}^{(k)}=\{[\lambda_{0},\dots,\lambda_{m-1}]\in \mathcal L^{(k)}_{m}|\#\{i\mid \lambda_{i}\neq\emptyset\}=\ell\}$ and
  $\mathcal R_{n,\ell}^{(k)}=\{[\rho_{0},\dots,\rho_{n-1}]\in \mathcal R^{(k)}_{n}|\#\{i\mid \rho_{i}\neq\emptyset\}=\ell\}$. Obviously 
  $\#\mathcal L_{m,\ell}^{(k)}=\#\mathcal R_{m,\ell}^{(k)}$. So we only investigate 
   $\#\mathcal R_{m,\ell}^{(k)}$. Let $P\in \mathcal R_{n,\ell}^{(k)}$.
    We define $\mathrm{succ}(P)=\{P\cup (P\cdot g)^{\uparrow}\mid   g\in  \IntEnt n^{ \IntEnt n}\}$ and $\mathrm{succ}_{\ell'}(P)=\{P\cup (P\cdot g)^{\uparrow}\in \mathcal R^{(k+1)}_{n,\ell'}\mid g\in  \IntEnt n^{ \IntEnt n}\}=\mathrm{succ}(P)\cap  \mathcal R^{(k+1)}_{n,\ell'}$. We notice that the cardinal of $\mathrm{succ}_{\ell'}(R)$ depends only on three parameters: $n$, $\ell$ and $\ell'$; we denote by $s_{n}^{(\ell,\ell')}$ this number. We remark that if $\ell'<\ell$ then  $s_{n}^{(\ell,\ell')}=0$.
    \begin{example}
	    Let $P=[\{1\},\{2\},\{3,4\},\emptyset,\emptyset]$. Then any $P'\in \mathrm{succ}_{4}(P)$ can be written as 
	    $P'=P\cup (P\cdot g)^{\uparrow}$ where $\#\left((\{0,1,2\}\cdot g)\cap \{3,4\}\right)=1$. In the aim to compute the number of elements of $\mathrm{succ}_{4}(P)$, one has only to consider the case where $(\{0,1,2\}\cdot g)\cap \{3,4\}=\{3\}$ and multiply this number by $2$; in the general case, one has to multiply by a binomial number. The set of functions $g$ satisfying $(\{0,1,2\}\cdot g)\cap \{3,4\}=\{3\}$  splits into several disjoint subsets according to the values of the number $\#\{i\in \{0,1,2\}\mid i\cdot g=3\}$.
	    \begin{enumerate}
		    \item For $\#\{i\in \{0,1,2\}\mid i\cdot g=3\}=3$, there exists only one function $g$ satisfying this condition and this function corresponds to the vector $[\{1\},\{2\},\{3,4\},\{5,6,7,8\},\emptyset]$.
		    \item For $\#\{i\in \{0,1,2\}\mid i\cdot g=3\}=2$, there are $3=\binom 32$ possibilities considering that  $\{i\in \{0,1,2\}\mid i\cdot g=3\}=\{0,1\}, \{0,2\}$ or $\{1,2\}$. For each of these possibilities one obtains $3$ vectors that correspond to the possible images of the only element which does not belong to $\{i\in \{0,1,2\}\mid i\cdot g=3\}$. So we obtain the following $9$ vectors: 
		    \[\begin{array}{ccc}\ [\{1,7,8\},\{2\},\{3,4\},\{5,6\},\emptyset],&[\{1\},\{2,7,8\},\{3,4\},\{5,6\},\emptyset],&[\{1\},\{2\},\{3,4,7,8\},\{5,6\},\emptyset],\\
		    \ [\{1,6\},\{2\},\{3,4\},\{5,7,8\},\emptyset],&[\{1\},\{2,6\},\{3,4\},\{5,7,8\},\emptyset],&[\{1\},\{2\},\{3,4,6\},\{5,7,8\},\emptyset],\\\ [\{1,5\},\{2\},\{3,4\},\{6,7,8\},\emptyset],&[\{1\},\{2,5\},\{3,4\},\{6,7,8\},\emptyset],&[\{1\},\{2\},\{3,4,5\},\{6,7,8\},\emptyset].\end{array}\]
		   \item  For $\#\{i\in \{0,1,2\}\mid i\cdot g=3\}=1$, there are $3=\binom 31$ possibilities considering that  $\{i\in \{0,1,2\}\mid i\cdot g=3\}=\{0\}, \{1\}$ or $\{2\}$. For each of these possibilities one obtains $9=3^{2}$ vectors that correspond to the possible images of the two elements which do not belong to $\{i\in \{0,1,2\}\mid i\cdot g=3\}$. So we obtain the following $27$ vectors: 
		   \[\begin{array}{ccc}\ [\{1,6,7,8\},\{2\},\{3,4\},\{5\},\emptyset],&[\{1,6\},\{2,7,8\},\{3,4\},\{5\},\emptyset],&[\{1,6\},\{2\},\{3,4,7,8\},\{5\},\emptyset],\\\ 
		   [\{1,7,8\},\{2,6\},\{3,4\},\{5\},\emptyset],& [\{1\},\{2,6,7,8\},\{3,4\},\{5\},\emptyset],& [\{1\},\{2,6\},\{3,4,7,8\},\{5\},\emptyset],\\\ 
		   [\{1,7,8\},\{2\},\{3,4,6\},\{5\},\emptyset],&[\{1\},\{2,7,8\},\{3,4,6\},\{5\},\emptyset],&[\{1\},\{2\},\{3,4,6,7,8\},\{5\},\emptyset],\\\ 
		   [\{1,5,7,8\},\{2\},\{3,4\},\{6\},\emptyset], &[\{1,5\},\{2,7,8\},\{3,4\},\{6\},\emptyset], &[\{1,5\},\{2\},\{3,4,7,8\},\{6\},\emptyset], \\\
		   [\{1,7,8\},\{2,5\},\{3,4\},\{6\},\emptyset],&[\{1\},\{2,5,7,8\},\{3,4\},\{6\},\emptyset], & [\{1\},\{2,5\},\{3,4,7,8\},\{6\},\emptyset],\\\ 
		   [\{1,7,8\},\{2\},\{3,4,5\},\{6\},\emptyset],& [\{1\},\{2,7,8\},\{3,4,5\},\{6\},\emptyset],& [\{1\},\{2\},\{3,4,5,7,8\},\{6\},\emptyset],\\\ 
		   [\{1,5,6\},\{2\},\{3,4\},\{7,8\},\emptyset], & [\{1,5\},\{2,6\},\{3,4\},\{7,8\},\emptyset], &[\{1,5\},\{2\},\{3,4,6\},\{7,8\},\emptyset],\\\ 
		   [\{1,6\},\{2,5\},\{3,4\},\{7,8\},\emptyset],& [\{1\},\{2,5,6\},\{3,4\},\{7,8\},\emptyset],& [\{1\},\{2,5\},\{3,4,6\},\{7,8\},\emptyset],\\\ 
		   [\{1,6\},\{2\},\{3,4,5\},\{7,8\},\emptyset],& [\{1\},\{2,6\},\{3,4,5\},\{7,8\},\emptyset],& [\{1\},\{2\},\{3,4,5,6\},\{7,8\},\emptyset].\end{array}\]
		\end{enumerate}  
		This gives  $s_{5}^{(3,4)}=1+9+27=37$.
	\end{example}
	This last example describes the strategy we use to obtain the following result.
    \begin{proposition}\label{values} We have
    \begin{equation}\label{Fs_lm}
	    s_{n}^{(\ell,\ell+\delta)}=\delta!\binom{n-\ell}{\delta}\sum_{\alpha=\delta}^{\ell}\binom{\ell}\alpha\ell^{\ell-\alpha}\left\{\alpha\atop \delta\right\},
	\end{equation}
	where $\left\{a\atop b\right\}$ denotes the Stirling number of second kind counting the number of partitions of $\{1,2,\dots,a\}$ into $b$ non empty sets.
    \end{proposition}
    \begin{proof}
	     Let	     
	    \begin{gather*}
	     P=[\rho_{0},\dots,\rho_{n-1}]\in\mathcal R_{n,\ell}^{(k)}\\
	      I=\{i_{1},\dots,i_{\ell}\}=\{i\mid P{i}\neq \emptyset\}\text{ and}\\
	       J=\{0,\dots,n-1\}\setminus I.
	       \end{gather*}
	        The set 
	       $\mathrm{succ}_{\ell+\delta}(P)$ splits as the following disjoint union
	       \begin{equation}
	       \mathrm{succ}_{\ell+\delta}(P)=\bigcup_{\alpha=\delta}^{\ell}\mathrm{succ}_{\ell+\delta,\alpha}(P),
	       \end{equation}
	       where
	       \begin{equation}
	       \mathrm{succ}_{\ell+\delta,\alpha}(P)=\left\{P\cup (P\cdot g)^\uparrow\in \mathrm{succ}_{\ell+\delta}(P)\mid \#\{i\in I\mid i\cdot g\in J\}=\alpha\right\}.
	       \end{equation}
	       But we have
	       \begin{equation}
	        \#\mathrm{succ}_{\ell+\delta,\alpha}(P)=\#\{g\in \IntEnt n^{I}\mid \#\{i\in I\mid i\cdot g\in J\}=\alpha\}.
	       \end{equation}
	        It remains to compute the cardinal  of $\chi_{\alpha}^{n,\ell,\delta}=\{g\in \IntEnt n^{I}\mid \#\{i\in I\mid i\cdot g\in J\}=\alpha\}$. Let $g\in \chi_{\alpha}^{n,\ell,\delta}$, then $I$ splits as the partition $I=I'\cup I''$ such that $(I'\cdot g)\subset I$, $(I''\cdot g)\subset J$, $\#(I''\cdot g)=\delta$, and $\#I''=\alpha$. To construct such a map, one has first to choose the sets $I''$ and $(I''\cdot g)$; we have $\binom\ell\alpha\binom{n-\ell}{\delta}$ possibilities to do that. Hence, we have to construct the image $(I''\cdot g)$ and this is equivalent to give an ordered partition of $I''$ into $\delta$ sets. We have $\delta!\left\{\alpha\atop \delta\right\}$ possibilities to do that.  Finally, to complete the description of $g$, we have to construct the restriction $g|_{I'}$ and this gives $\ell^{\ell-\alpha}$ possibilities. In conclusion, we have shown that $s_{n}^{(\ell,\ell+\delta)}=\sum_{\alpha=\delta}^{\ell}\# \chi_{\alpha}^{n,\ell,\delta}$ and $\# \chi_{\alpha}^{n,\ell,\delta}=\delta!\binom{n-\ell}{\delta}\binom\ell\alpha\ell^{\ell-\alpha}\left\{\alpha\atop\delta\right\}$. This is equivalent to our statement.\cqfd
	\end{proof}
	\begin{example}
		Let us  illustrate where the Stirling numbers in formula (\ref{Fs_lm}) come from. In fact, one has to compute the number of surjective functions from a set having $\alpha$ elements onto a set having $\delta$ elements. For our example, we suppose that we want to enumerate the surjective functions from $\{0,1,2,3,4\}$ to $\{0,1,2\}$. Let $g$ be such a map. The image of the elements by $g$ partitions the initial set into $3$ disjoint parts  $\displaystyle\{0,1,2,3,4\}=\bigcup_{i=0}^{2}\{j\mid f(j)=i\}$. The number of such set partitions equals to the Stirling number $\left\{5\atop 3\right\}=25$. But a partition does not characterizes completely a surjective function and we need also to take into account the image associated to each part. So  we have to multiply the number  $\left\{5\atop 3\right\}$ by $\delta!=6$ to obtain the number of surjections (here $150$). This gives the factors $\delta!\left\{\alpha\atop\delta\right\}$ in each term of  (\ref{Fs_lm}).
	\end{example}
	\begin{example}
		The numbers $s_{n}^{(i,j)}$ have some special values. Let us list some of them.
		\begin{itemize}
			\item $s_{n}^{(i,i)}=i^{i}$ for $1\leq i\leq n$,
			\item $s_{n}^{(i,i+1)}=(n-i)((i+1)^{i}-(i)^{i})$ for $1\leq i\leq n-1$,
			\item $s_{n}^{(i,2i)}=i!\binom {n-i}i$ for $1\leq i\leq \frac n2$, in particular  $s_{2n}^{(n,2n)}=n!$.
		\end{itemize}
	\end{example}
	Recall that the Hadamard product of two matrices of the same dimension $M=(m_{ij})_{ij}$ and $N=(n_{ij})_{ij}$ is the matrix $M\cdot N=(m_{ij}n_{ij})_{ij}$. Another way to state Proposition  \ref{values} is to write that the coefficient $s_{n}^{(i,j)}$ is an entry of an infinite matrix \begin{equation}\mathbf S_{n}=\mathbf B\cdot\mathbf A_{n}\end{equation} that is the Hadamard product of the matrix 
	\begin{equation}\label{MatriceB}\mathbf B=\left(\sum_{k=j-i}^{i}\binom iki^{i-k}\left\{k\atop j-i\right\}\right)_{ij}\end{equation}
	 which does not depend on $n$ with  the matrix $\mathbf A_{n}=\left((j-i)!\binom {n-i} {j-i}\right)_{ij}$ depending on $n$.
	 \begin{example}
		 \[{}
		 \left[{}
		 \begin{array}{cccccccc}
			 1&1&0&0&0&0&0&\dots{}\\
			 0&4&5&1&0&0&0&\dots\\
			 0&0&27&37&12&1&0&\dots\\
			 0&0&0&256&369&151&22&\dots\\
			 0&0&0&0&3125&4651&2190&\dots\\
			 0&0&0&0&0&46656&70993&\dots\\
			 \vdots&\vdots&\vdots&\vdots&\vdots&\vdots&\vdots
		 \end{array}
		 \right]\cdot
		 \left[{}
		 \begin{array}{cccccccc}
			 1&4&12&24&24&0&0&\dots{}\\
			 0&1&3&6&6&0&0&\dots\\
			 0&0&1&2&2&0&0&\dots\\
			 0&0&0&1&1&0&0&\dots\\
			 0&0&0&0&1&0&0&\dots\\
			 0&0&0&0&0&0&0&\dots\\
			 \vdots&\vdots&\vdots&\vdots&\vdots&\vdots&\vdots
		 \end{array}
		 \right]={}
		 \left[{}
		 \begin{array}{cccccccc}
			 1&4&0&0&0&0&0&\dots{}\\
			 0&4&15&6&0&0&0&\dots\\
			 0&0&27&74&24&0&0&\dots\\
			 0&0&0&256&369&0&0&\dots\\
			 0&0&0&0&3125&0&0&\dots\\
			 0&0&0&0&0&0&0&\dots\\
			 \vdots&\vdots&\vdots&\vdots&\vdots&\vdots&\vdots
		 \end{array}
		 \right].
		 \]
	 \end{example}
	 These matrices have combinatorial interpretation. First the entry $(i,j)$ of $\mathbf A_{n}$ is nothing but the number $a_{n-i,j-i}={(n-i)!\over (j-i)!}$ of ways of obtaining an ordered subset of $j-i$ elements from a set of $n-i$ elements (in the case where $n-i<0$, $j-i<0$ or $n<j$ we assume $a_{n-i,j-i}=0$ by convention).\\
	 The entries of the matrix $\mathbf B$ are interpreted  in terms of $r$-Stirling numbers. The $r$-Stirling number  $\left\{n\atop k\right\}_{r}$ is the number of partitions of a set of $n$ elements into $k$ nonempty disjoints subsets such that the first $r$ elements are in distinct subsets \cite{Broder1984}. The $r$-Bell polynomials \cite{Mezo2010} are defined by 
	 \begin{equation}B_{n,r}(x)=\sum_{k=0}^{n}\left\{n+r\atop k+r\right\}_{r}x^{r}.
	 \end{equation}
	 Their exponential generating function is (see \cite{Mezo2010} Theorem 3.1)
	 \begin{equation}\label{genfuncBr}
		 \sum_{n}B_{n,r}{x^{n}\over n!}=e^{x(e^{z}-1)+rz},
	 \end{equation}
	 and they satisfy
	 (see \cite{Mezo2010} Corollary 3.2)
	 \begin{equation}\label{cor32}
		 B_{n,r}(x)=\sum_{k=0}^{n}r^{k}\binom nkB_{n-k}(x),
	 \end{equation} 
	 where $B_{n}(x)=\sum_{k=0}^{n}\left\{n\atop k\right\}x^{k}$ is the usual Bell polynomial. Comparing (\ref{MatriceB}) and (\ref{cor32}), we find
	 \begin{equation}
		 \mathbf B_{i,j}=\left\{2i\atop j\right\}_{i}.
	 \end{equation}
	 For instance, the numbers of the line $i=2$ are interpreted as follows: 
	 \begin{itemize}
		 \item There are $4$ partitions of $\{1,2,3,4\}$ into two parts such that the numbers $1$ and $2$ are in two distinct parts: $\{\{1,3,4\},\{2\}\}$, $\{\{1,4\},\{2,3\}\}$, $\{\{1,3\},\{2,4\}\}$, and  $\{\{1\},\{2,3,4\}\}$.
		 \item There are $5$ partitions of $\{1,2,3,4\}$ into three parts such that the numbers $1$, $2$, and $3$ are in three distinct parts: $\{\{1,4\},\{2\},\{3\}\}$, $\{\{1\},\{2,4\},\{3\}\}$,  $\{\{1\},\{2\},\{3,4\}\}$, $\{\{1,3\},\{2\},\{4\}\}$, and $\{\{1\},\{2,3\},\{4\}\}$.
		 \item There is only $1$ partition of $\{1,2,3,4\}$ into four parts such that the numbers $1$, $2$, $3$, and $4$ are in four distinct parts.
	\end{itemize}
	So we have the following result.
	\begin{proposition}
		\begin{equation}
			s^{(i,j)}_{n}=a_{n-i,j-i}\left\{2i\atop j\right\}_{i}.
		\end{equation}
	\end{proposition}
	\subsection{Generating functions}
	From (\ref{genfuncBr}), we deduce that the coefficient of $x^{j}y^{i}z^{i}$ in the Taylor expansion of $e^{x(e^{z}-1)}\over 1-xye^{z}$ equals $\frac1{i!}\mathbf B_{i,j}$. Equivalently, consider the function  $e^{x(e^{z}-1)}z\over z-xye^{z}$ as a Taylor series in $y$. Each monomial $y^{i}$ has a coefficients that is a Laurent series in $z$. Hence, the double generating function of $\frac1{i!}\mathbf B_{i,j}$ is obtained by computing the constant term in $z$ in this expansion. In terms of residue, we summarize this result as
	\begin{equation}
			\sum_{i,j}B_{i,j}x^{j}{y^{i}\over i!}=\mathrm {Res}_{z=0}\left({e^{x(e^{z}-1)}\over z-xye^{z}}\right){},
		\end{equation}
		where $\mathrm {Res}_{z=0}$ denotes the residue in $z=0$ acting on each terms of the Taylor expansion in $y$  of its argument.
	\begin{example}
		We have
		\[{}
		{e^{x(e^{z}-1)}\over z-xye^{z}}=\sum_{i=0}^{\infty}z^{-i-1}e^{x(e^{z}-1)}e^{iz}y^{i}.
		\]
		The coefficient of $y^{4}$ in the previous expression admits the following Laurent expansion
		\[{}
		\begin{array}{rcl}
		z^{-5}e^{x(e^{z}-1)}e^{4z}&=&x^{4}z^{-5}+x^{4}(x+4)z^{-4}+\frac12x^{4}(16+9x+x^{2})z^{-3}+\frac16x^{4}(64+61x+15x^{2}+x^{3)}z^{-2}\\&&+\frac1{24}x^{4}(256+369x+151x^{2}+22x^{3}+x^{4})z^{-1}\\&&+\frac1{120}(1024+2101x+1275x^{2}+305x^{3}+30x^{4}+x^{5})+\cdots{}
		\end{array}
		\]
		The residue is nothing but the coefficient of $z^{-1}$ in this expression, that is $\frac1{24}x^{4}(256+369x+151x^{2}+22x^{3}+x^{4})$ as expected.
	\end{example}
		But the residue Theorem implies
		\begin{equation}
			\mathrm {Res}_{z=0}\left({e^{x(e^{z}-1)}\over z-xye^{z}}\right)=\frac1{2\mathbf i\pi}\oint_{\gamma}{e^{x(e^{z}-1)}\over z-xye^{z}}dz
		\end{equation}
		where $\gamma$ is a counterclockwise path  around a Jordan curve enclosing $0$ and $\mathbf i^{2}=-1$. In terms of symbolic computations, all works as if the curve encloses all the poles of  ${e^{x(e^{z}-1)}\over z-xye^{z}}$ now considered as a function of $z$. The function $z\rightarrow z-xye^{z}$ admits only one (simple) zero  at the value $z=-\mathrm W(-xy)$, where
		  $\mathrm W$ denotes the Lambert $W$ function \cite{Lambert1758,Euler1783}. We recall that $W(z)$ is the inverse function of $z\rightarrow ze^{z}$ and its Taylor expansion  (see \emph{eg} \cite{CGHJK1996}) is
		  \begin{equation}
			 \mathrm W(z)=\sum_{i>0}(-i)^{i-1}{z^{i}\over i!}.
		  \end{equation}
		  Hence,
		  \begin{equation}
			  \sum_{i,j}B_{i,j}x^{j}{y^{i}\over i!}=\mathrm {Res}_{z=-\mathrm W(-xy)}\left({e^{x(e^{z}-1)}\over z-xye^{z}}\right).
		  \end{equation}
		  Since $z=-\mathrm W(-xy)$ is the only pole of ${e^{x(e^{z}-1)}\over z-xye^{z}}$ of order $1$, the series  $ \sum_{i,j}B_{i,j}x^{j}{y^{i}\over i!}$ is the constant term in the Taylor expansion of $(z+\mathrm W(-xy)){e^{x(e^{z}-1)}\over z-xye^{z}}$ at $z=-\mathrm W(-xy)$. In other words,
		  \begin{equation}\label{B2W}
			  \sum_{i,j}B_{i,j}x^{j}{y^{i}\over i!}=\lim_{z\rightarrow -\mathrm W(-xy)}(z+\mathrm W(-xy)){e^{x(e^{z}-1)}\over z-xye^{z}}.
		  \end{equation}
		  But noticing that $\mathrm W(a)=ae^{-\mathrm W(a)}$ and $\mathrm W'(a)={\mathrm W(a)\over a(1+\mathrm W(a))}$, we obtain
		  \begin{equation}
			  \lim_{\alpha\rightarrow\mathrm W(a)}{\alpha-ae^{-\alpha}\over \alpha-W(a)}=\lim_{\beta\rightarrow a}{\mathrm W(\beta)-ae^{-\mathrm W(\beta)}\over \mathrm W(\beta)-\mathrm W(a)}=\lim_{\beta\rightarrow a}{(\beta-a)\mathrm W(a)\over (\mathrm W(\beta)-\mathrm W(a))\beta}={\mathrm W(a)\over a\mathrm W'(a)}=1+\mathrm W(a).
		  \end{equation}
		  Using this equality in (\ref{B2W}) we find the following result.
	\begin{proposition}
		The generating series of the coefficients $B_{i,j}$ is
		\begin{equation}
			\sum_{i,j}B_{i,j}x^{j}{y^{i}\over i!}={e^{-({\mathrm W(-xy)\over y}+x)}\over 1+\mathrm W(-xy)}.
		\end{equation}
	\end{proposition}
	\begin{example}
		By applying the previous proposition, we find:
		\begin{equation}\begin{array}{rcl}\displaystyle
			\sum_{i,j}B_{i,j}x^{j}{y^{i}\over i!}&=&1+x \left( 1+x \right) y+\frac12\,{x}^{2} \left( x+4 \right)  \left( 1+x
 \right) {y}^{2}+\frac16\,{x}^{3} \left( 27+37\,x+12\,{x}^{2}+{x}^{3}
 \right) {y}^{3}\\&&+\frac1{24}\,{x}^{4} \left( 256+369\,x+151\,{x}^{2}+22\,{x}^
{3}+{x}^{4} \right) {y}^{4}+\cdots.\end{array}
		\end{equation}
	\end{example}
\subsection{Closed expressions for the number of successors}	
	Define the matrix $S_{n}$ constituted with the $n$ first  lines and the $n$ first columns in $\mathbf S$, that is $S_{n}=(s_{n}^{(i,j)})_{1\leq i,j\leq n}$. 
	\begin{example}
		The first matrices $S_{n}$ follow:
		\begin{equation}
			S_{2}=\left[ \begin {array}{cc} 1&1\\ \noalign{\medskip}0&4\end {array}
 \right],\ S_{3}=  \left[ \begin {array}{ccc} 1&2&0\\ \noalign{\medskip}0&4&5
\\ \noalign{\medskip}0&0&27\end {array} \right], S_{4}=\left[ \begin {array}{cccc} 1&3&0&0\\ \noalign{\medskip}0&4&10&2
\\ \noalign{\medskip}0&0&27&37\\ \noalign{\medskip}0&0&0&256
\end {array} \right], S_{5}= \left[ \begin {array}{ccccc} 1&4&0&0&0\\ \noalign{\medskip}0&4&15&6&0
\\ \noalign{\medskip}0&0&27&74&24\\ \noalign{\medskip}0&0&0&256&369
\\ \noalign{\medskip}0&0&0&0&3125\end {array} \right] ,
		\end{equation}
		\begin{equation}
			S_{6}= \left[ \begin {array}{cccccc} 1&5&0&0&0&0\\ \noalign{\medskip}0&4&20&
12&0&0\\ \noalign{\medskip}0&0&27&111&72&6\\ \noalign{\medskip}0&0&0&
256&738&302\\ \noalign{\medskip}0&0&0&0&3125&4651\\ \noalign{\medskip}0
&0&0&0&0&46656\end {array} \right], S_{7}= \left[ \begin {array}{ccccccc} 1&6&0&0&0&0&0\\ \noalign{\medskip}0&4&
25&20&0&0&0\\ \noalign{\medskip}0&0&27&148&144&24&0
\\ \noalign{\medskip}0&0&0&256&1107&906&132\\ \noalign{\medskip}0&0&0&0
&3125&9302&4380\\ \noalign{\medskip}0&0&0&0&0&46656&70993
\\ \noalign{\medskip}0&0&0&0&0&0&823543\end {array} \right] ,\dots
		\end{equation}
	\end{example}
		
Let us denote by $s_{n,k}^{(i,j)}$ the entries of the $k$th power $(S_{n})^{k}$.
From its definition, the number  $s_{n,k}^{(i,j)}$ is the cardinal of the set $\mathrm{succ}^{k}(P)\cap \mathcal R^{(k+k')}_{n,j}$ for any $P\in \mathcal R^{(k')}_{n,i}$. Hence, the entries  $s_{n,k}^{(1,j)}$ of the first line of $(S_{n})^{k}$ are the cardinal of the sets of $\mathcal R^{(k)}_{n,j}$ obtained by applying $k$ times the map $\mathrm succ$ to the vector $[\{1\},\emptyset,\dots,\emptyset]$. Equivalently, we have the following statement.
	\begin{proposition}
		We have  \begin{equation}\label{Lkl}\#\mathcal R_{n,\ell}^{(k)}=s_{n,k}^{(1,\ell)},\end{equation}
		 \begin{equation}\#\mathcal R^{(k)}_{m}=\sum_{\ell}s_{n,k}^{(1,\ell)},\end{equation} and
		 \begin{equation} 
		\#\mathcal U_{m,n}^{(k)}=\sum_{\ell',\ell''}s_{m,k}^{(1,\ell')}s_{n,k}^{(1,\ell'')}.\end{equation}
	\end{proposition}
	
\begin{example}
	Examine the first line of $S_{3}^{2}= \left[ \begin {array}{ccc} 1&10&10\\ \noalign{\medskip}0&16&155
\\ \noalign{\medskip}0&0&729\end {array} \right]$. This means that $\mathcal R^{(2)}_{3}$ contains
\begin{enumerate}
	\item One vector having exactly one non empty entry: $[\{1,2,3,4\},\emptyset,\emptyset]$;{}
	\item Ten vectors having exactly two non empty entries: $[\{1,2\},\{3,4\},\emptyset]$, $[\{1,2\},\emptyset,\{3,4\}]$, 
	$[\{1,3,4\},\{2\},\emptyset]$,  $[\{1,3\},\{2,4\},\emptyset]$, $[\{1,4\},\{2,3\},\emptyset]$, $[\{1\},\{2,3,4\},\emptyset]$,{}
	$[\{1,3,4\},\emptyset,\{2\}]$,  $[\{1,3\},\emptyset,\{2,4\}]$, $[\{1,4\},\emptyset,\{2,3\}]$, $[\{1\},\emptyset,\{2,3,4\}]$;
	\item Ten vectors having exactly three non empty entries: $[\{1\},\{2\},\{3,4\}]$, $[\{1,3\},\{2\},\{4\}]$, 
	$[\{1\},\{2,4\},\{4\}]$, $[\{1,4\},\{2\},\{3,4\}]$, $[\{1\},\{2,4\},\{3\}]$,
	$[\{1\},\{3,4\},\{2\}]$, $[\{1,3\},\{4\},\{2\}]$, 
	$[\{1\},\{4\},\{2,4\}]$, $[\{1,4\},\{3,4\},\{2\}]$, $[\{1\},\{3\},\{2,4\}].$ 
\end{enumerate} 
\end{example}
\begin{remark}
	Noticing that the diagonal entries in the matrices $S_{n}$ are $1^{1}$, $2^{2},\ \dots,$ $m^{m}$. It results that the number of elements $\mathcal R_{n}^{(k)}$ is a linear combination
	\begin{equation}
		\#\mathcal R_{n}^{(k)}=\sum_{i=1}^{n}a_{i}^{(n)}(i^{i})^{k},
	\end{equation}
	where $a_{i}^{(n)}$ are rational numbers that are independent of $k$. While we do not know a closed form for the values of the  $a_{i}^{(n)}$ coefficients, we can easily compute them by solving a system of linear equations or, alternatively, use some known formulas as those described in \cite{WShur2011}.\\
	The first formulas are:
	\begin{eqnarray}
		\#\mathcal R^{(k)}_{2}&=&\frac23+\frac13\,{4}^{k},\\
		\#\mathcal R^{(k)}_{3}&=&{\frac {6}{13}}+{\frac {12}{23}}\,{4}^{k}+{\frac {5}{299}}\,{27}^{k}
,\\
\#\mathcal R^{(k)}_{4}&=&{\frac {372}{1105}}+{\frac {100}{161}}\,{4}^{k}+{\frac {2880}{68471}}
\,{27}^{k}+{\frac {23}{136255}}\,{256}^{k},\\
\#\mathcal R^{(k)}_{5}&=&{\frac {135040}{517803}}+{\frac {1006400}{1507443}}\,{4}^{k}+{\frac {
7530000}{106061579}}\,{27}^{k}+{\frac {46000}{78183119}}\,{256}^{k}+{
\frac {4150701}{10832451881581}}\,{3125}^{k},\\
\#\mathcal R^{(k)}_{6}&=&{\frac {344810430}{1610539931}}+{\frac {4008890625}{5860435903}}\,{4}^
{k}+{\frac {18418610000}{183168346933}}\,{27}^{k}+{\frac {5833053}{
4534620902}}\,{256}^{k}\nonumber\\&&+{\frac {806896274400}{471547462857102511}}\,{
3125}^{k}+{\frac {114196541}{474523188718486138}}\,{46656}^{k},\\
\#\mathcal R^{(k)}_{7}&=&{\frac {5818082250876}{31579697044181}}+{\frac {157292430099924}{
229823691577177}}\,{4}^{k}+{\frac {4868336034090900}{37710516098219107
}}\,{27}^{k}+{\frac {200723945058}{88887962822497}}\,{256}^{k}\nonumber\\&&+{\frac 
{888824849603838210}{193433013191149163934799}}\,{3125}^{k}+{\frac {
34547422762566}{26332206893852752878029}}\,{46656}^{k}\nonumber\\&&+{\frac {
32920001103738912355}{678426042037319159866567474314373}}\,{823543}^{k
},
\dots
\end{eqnarray}
All the coefficients seem to be positive but their combinatorial interpretation remains to be investigated.
It is also interesting to note that the coefficient $a_{n}^{(n)}$ of the dominant power $(n^{n})^{k}$ in $\#\mathcal R_{n}^{(k)}$ decreases very quickly to zero. For instance, $a_{35}^{(35)}<{ 5.267408697\cdot10^{-238}}$.
\end{remark}
The first values of $\#\mathcal R_{n}^{(k)}$ for $n=2$ are  $1,\ 2,\ 6,\ 22,\ 86,\ 342,\ 1366,\ 5462,\ 21846,\ 87382,\ 349526,\dots$. This series of number appears in other contexts. The sequence A047849 of \cite{Sloane} lists some of other interpretations. For instance, these numbers also count the closed walks of length $2n$ at a vertex of the cyclic graph on $6$ nodes, the permutations of length $n$ avoiding $4321$ and $4123$, the closed walks of length $n$ at a vertex of a triangle with two loops at each vertex etc. The sequences for other value of $n$ are not referenced in \cite{Sloane}. It should be interesting to investigate if some of these interpretations naturally extends for $n>3$. By construction, one notices that, for any $k\geq 1$, $n$ divides $\#\mathcal R_{n}^{(k)}$.

 \section{The state complexity of the shuffle operation revisited}\label{sect-sc}
 Let $(\Lambda,P)\in\mathcal U_{m,n}$ with $\Lambda=[\lambda_{0},\dots,\lambda_{m-1}]$ and $P=[\rho_{0},\dots,\rho_{n-1}]$. We define $\mathfrak s(\Lambda,P)=\{(i,j)\mid \lambda_{i}\cap\rho_{j}\neq\emptyset\}$. 
 \begin{proposition}
	 The set of accessible states of  $\mathfrak{Shuf}(M^{n_1}_{F_1},M^{n_2}_{F_2})$ is 
	 $$\left\{s(\Lambda,P)\mid \exists k\leq \mathbf f(m,n), (\Lambda,P)\in \mathcal U^{(k)}_{m,n}\right\}.$$
 \end{proposition}

  \begin{proof} Let us  prove by induction that $\mathfrak{Shuf}(M^{n_1}_{F_1},M^{n_2}_{F_2})\subset \left\{s(\Lambda,P)\mid \exists k\leq \mathbf f(m,n), (\Lambda,P)\in \mathcal U^{(k)}_{m,n}\right\}$. We  have $\mathfrak s(\mathcal P(()))=\{(0,0)\}$. Suppose that $\mathfrak p=(\mathfrak t_{1},\dots,\mathfrak t_{k})$ is 
  a path from $\{(0,0)\}$ to $E$ and assume, as an induction hypothesis, that  $\mathfrak s(\mathcal P(\mathfrak p))=E$.  Let $\mathfrak t_{k+1}=(E,(f,g),E')$ be a transition and set $\mathfrak p'=(\mathfrak t_{1},\dots,\mathfrak t_{k+1})$. Then $E'=\{(i\cdot f,j)\mid (i,j)\in E\}\cup \{(i,j\cdot g)\mid (i,j)\in E\}$. If we set $\mathcal P(\mathfrak p)=(\Lambda,P)$ with $\Lambda=[\lambda_{0},\dots,\lambda_{m-1}]$ and $P=[\rho_{0},\dots,\rho_{n-1}]$, then we have
  \begin{equation}
  \mathcal P(\mathfrak p')=((\Lambda\cdot f)\cup\Lambda^{\uparrow},P\cup (P\cdot g)^{\uparrow}).
  \end{equation}
  Hence, if we set $(\Lambda\cdot f)\cup\Lambda^{\uparrow}=[\lambda'_{0},\dots,\lambda'_{m-1}]$ and $P\cup (P\cdot {g})^{\uparrow}=[\rho'_{0},\dots,\rho'_{n-1}]$ we have
	 \begin{eqnarray}
	 \mathfrak s(\mathcal P(\mathfrak p'))&=&\{(i,j)\mid \lambda'_{i}\cap\rho'_{j}\neq \emptyset\}\nonumber\\&=&
  \{(i\cdot f,j)\mid \lambda_{i}\cap\rho_{j}\neq \emptyset\}\cup \{(i,j\cdot g)\mid \lambda_{i}\cap\rho_{j}\neq \emptyset\}\\
  &=&\{(i\cdot f,j)\mid (i,j)\in E\}\cup \{(i,j\cdot g)\mid (i,j)\in E\}=E',\nonumber
  \end{eqnarray}
  as expected.
  
  Conversely, by construction, for any state $E$ there exists $(\Lambda,P)\in\mathcal U_{m,n}$ such that $\mathfrak s(\Lambda,P)=E$. Since the state complexity is bounded by $\mathbf f(m,n)$, any state is reached in at most $\mathbf f(m,n)$ transitions in $\mathfrak{Shuf}(M_{F_{1}}^{m},M_{F_{2}}^{n})$. This means that there exists $k\leq \mathbf f(m,n)$ such that there exists  $(\Lambda,P)\in\mathcal U_{m,n}^{(k)}$ with $\mathfrak s(\Lambda,P)=E$. \cqfd
 \end{proof}
 From all the results of this section, one deduces
 \begin{theorem}
	 \begin{equation}\label{ScSh}
		 sc_{\shuffle}(m,n)=\#\left\{\mathfrak s(\Lambda,P)\mid (\Lambda,P)\in\bigcup_{k=0}^{\mathbf f(m,n)}\mathcal U^{(k)}_{m,n}\right\}.
	 \end{equation}
 \end{theorem}
 Although Formula (\ref{ScSh}) is an exact expression for the state complexity of the shuffle product, this is not a number easy to manipulate because the set $\bigcup_{k\leq\mathbf f(m,n)}\mathcal U^{(k)}_{m,n}$ is very huge. In fact, we would hope to find
 $\#\left\{\mathfrak s(\Lambda,P)\mid (\Lambda,P)\in\bigcup_{k=0}^{\mathbf f(m,n)}\mathcal U^{(k)}_{m,n}\right\}=\mathbf f(m,n)$. Equivalently
 the conjecture of Brzozowski \emph{et al.} reduces to
 \begin{conj}\label{conj1}
	 \begin{equation}
		\mathfrak s\left(\mathcal U_{m,n}\right)=
		\left\{E\subset \IntEnt m\times \IntEnt  n\mid{}
		  E\cap (\{0\}\times \IntEnt n)\neq \emptyset \text{ and }
		    E\cap (\IntEnt m\times\{0\})\neq \emptyset\right\}
	 \end{equation}
 \end{conj} 
 \begin{example}
	 For $m=n=2$, the $\mathbf f(2,2)=10$ states are recovered from $\mathcal U_{2,2}$ as follows
	 $$
	 \begin{array}{cccccccc}
		 \begin{array}{|c|c|}\hline
		 \times&\color{white}\times \\\hline &\color{white}x \\\hline
	 \end{array}=&\mathfrak s\left(\ \begin{array}{|c|c|}\hline
		 1&\color{white}1 \\\hline\ &\ \\\hline 
	 \end{array}\ \right),\\ \\
	 \begin{array}{|c|c|}\hline
		 \times&\times \\\hline  &\color{white}x\\\hline 
	 \end{array}=&\mathfrak s\left(\ \begin{array}{|c|c|}\hline
		 1&2 \\\hline \color{white}x&\ \\\hline 
	 \end{array}\ \right),&{}
	  \begin{array}{|c|c|}\hline
		 \times& \\\hline  \times&\color{white}x\\\hline 
	 \end{array}=&\mathfrak s\left(\ \begin{array}{|c|c|}\hline
		 2&\color{white}2 \\\hline 1&\ \\\hline  
	 \end{array}\ \right),& 
		 \begin{array}{|c|c|}\hline
		 \color{white}\times&\times \\\hline \times &\color{white}x \\\hline
	 \end{array}=&\mathfrak s\left(\ \begin{array}{|c|c|}\hline
		 &2 \\\hline 1&\  \\\hline
	 \end{array}\ \right),&{}
	 \begin{array}{|c|c|}\hline
		 \times& \\\hline  &\times\\\hline 
	 \end{array}=&\mathfrak s\left(\ \begin{array}{|c|c|}\hline
		 1\,4& \\\hline &2\,3 \\\hline  
	 \end{array}\ \right),\\ \\
	 \begin{array}{|c|c|}\hline
		 \times&\times\\\hline \times&\ \\\hline 
	 \end{array}=&\mathfrak s\left(\ \begin{array}{|c|c|} \hline
		 1&2\,4 \\\hline 3 &  \\\hline
	 \end{array}\ \right),&\ 
		 \begin{array}{|c|c|} \hline
		 \times& \\\hline \times &\times\\\hline 
	 \end{array}=&\mathfrak s\left(\ \begin{array}{|c|c|}\hline
		 1& 4\\\hline &2\,3  \\\hline
	 \end{array}\ \right),&{}
	 \begin{array}{|c|c|}\hline
		 \times&\times \\\hline  &\times\\\hline 
	 \end{array}=&\mathfrak s\left(\ \begin{array}{|c|c|}\hline
		 1\,4& \\\hline 3&2\\  \hline
	 \end{array}\ \right),&
	 \begin{array}{|c|c|}\hline
		 &\times\\\hline \times&\times\ \\\hline 
	 \end{array}=&\mathfrak s\left(\ \begin{array}{|c|c|}\hline
		 &4 \\\hline 1\,3 & 2 \\\hline
	 \end{array}\ \right),\\ \\
		 \begin{array}{|c|c|}
		 \hline\times&\times \\\hline \times &\times \\\hline
	 \end{array}=&\mathfrak s\left(\ \begin{array}{|c|c|}\hline
		 1& 4\\\hline 3&2  \\\hline
	 \end{array}\ \right).
	 \end{array}
	 $$
 \end{example}

Although we have no general algorithm allowing us to compute a reverse function to $\mathfrak s$, we know how to handle a few families of tableaux.
\begin{example}\label{permutations}
	Suppose that $n=m$ and define $E_{\sigma}=\{(i,i\cdot \sigma)\mid i\in\IntEnt n\}$ for any permutation $\sigma$ of $\IntEnt n$.
	First suppose that $0\cdot \sigma=0$ and let $k$ be such that $2^{k-1}\leq n< 2^{k}$. We consider the vector $\Pi_{n}=[\{1,2^{k}\},\{2,2^{k}-1\},\dots,\{2^{k}-n,n+1\},\{2^{k}-n+1\},\{2^{k}-n+2\},\dots,\{n\}]$. We check that $\Pi_{n}\in\mathcal L^{(k)}_{n}\cap\mathcal R^{(k)}_{n}$ and $\mathfrak s([\Pi_{n},\Pi_{n}])=\{(i,i)\mid i\in\IntEnt n\}$. For instance for $n=5$, one has
	\begin{equation}
		\mathfrak s\left(\ 
		\begin{array}{|c|c|c|c|c|}\hline
			1\ 8&&&&\\\hline
			&2\ 7&&&\\\hline
			&&3\ 6&&\\\hline
			&&&4&\\\hline
			&&&&5\ \\\hline
		\end{array}\ 
		\right)=\begin{array}{|c|c|c|c|c|}\hline
			\times&&&&\\\hline
			&\times&&&\\\hline
			&&\times&&\\\hline
			&&&\times&\\\hline
			&&&&\times\\\hline
		\end{array}
	\end{equation}
	To construct any other set $E_{\sigma}$ with $0\cdot\sigma=0$ it suffices to consider 
	the pair $[\Pi_{n}\cdot \sigma,\Pi_{n}]$  because we have $(\Pi_{n}\cdot \sigma)[i]=\Pi_{n}[i\cdot \sigma]$.{}
	For instance we have,
	\begin{equation}
		\mathfrak s\left(\ 
		\begin{array}{|c|c|c|c|c|}\hline
			1\ 8&&&&\\\hline
			&&3\ 6&&\\\hline
			&&&&5\\\hline
			&&&4&\\\hline
			&2\ 7&&&\ \\\hline
		\end{array}\ 
		\right)=\begin{array}{|c|c|c|c|c|}\hline
			\times&&&&\\\hline
			&&\times&&\\\hline
			&&&&\times\\\hline
			&&&\times&\\\hline
			&\times&&&\\\hline
		\end{array}
	\end{equation}	
Now suppose that $0\cdot\sigma\neq 0$. Let $k$ such that $2^{k-1}<n\leq 2^{k}$ 
and $r=2n-2^{k}$. If $n$ is even we set $q=2^{k-1}-\frac n2$ and
$\Pi_{n}=[\{1,2^{k-1}+1\},\dots,\{q,2^{k-1}+q\},\{2^{k-1},2^{k}\},\{2^{k-1}-1,2^{k}-1\},\dots,\{2^{k-1}-q+1,2^{k}-q+1\},\{q+1\},\dots,\{2^{k-1}-q\},\{2^{k-1}+q+1\},\dots,\{2^{k}-q\}]$. This vector partitions the set $\{1,\dots,2^{k}\}$ into $n=2q+r$ parts with $2q$ entries of size $2$ and $2(2^{k-1}-2q)=r$ entries of size $1$. If $n$ is odd we set $q=2^{k-1}-\frac{n-1}2$ and $\Pi_{n}=
[\{1,2^{k-1}+1\},\dots,\{q,2^{k-1}+q\},\{2^{k-1},2^{k}\},\{2^{k-1}-1,2^{k}-1\},\dots,\{2^{k-1}-q+2,2^{k}-q+2\},\{q+1\},\dots,\{2^{k-1}-q+1\},\{2^{k-1}+q+1\},\dots,\{2^{k}-q+1\}]$. This vector partitions the set  $\{1,\dots,2^{k}\}$ into $n=2q-1+r$  parts with $2q-1$ entries of size $2$ and $2(2^{k-1}-q+1-(q+1)+1)=r$ entries of size $1$. In both cases ($n$ even or $n$ odd), a permutation $\Lambda$ (resp. $P$) of $\Pi_{n}$ belongs to $\mathcal L_{n}^{(k)}$ (resp. to $\mathcal R_{n}^{(k)}$) if and only if $\Lambda[0]=\{2^{k-1},2^{k}\}$ (resp. $P[0]=\{1,2^{k-1}+1\}$). To construct a pair $[\Lambda,P]$ such that $\mathfrak s([\Lambda,P])=E_{\sigma}$ we proceed as follows
\begin{enumerate}
	\item We set $\Lambda[0]=P[0\cdot\sigma ]=\{2^{k-1},2^{k}\}$ and $\Lambda[0\cdot \sigma^{-1}]=P[0]=\{1,2^{k-1}+1\}$.
	\item We choose randomly the remaining entries in $\Lambda$ in such a way that $\Lambda$ is a permutation of $\Pi_{n}$.
	\item If $j$ is the index of an entries filled in the previous step, we set $P[j\cdot\sigma]=\Lambda[j]$.
\end{enumerate}
For instance, consider the following tableau for $n=5$
\begin{equation}
E_{(132)(45)}=
\begin{array}{|c|c|c|c|c|}\hline
			&&\times&&\\\hline
			\times&&&&\\\hline
			&\times&&&\\\hline
			&&&&\times\\\hline
			&&&\times&\\\hline
		\end{array}
\end{equation}
We have $\Pi_{5}=[\{1,5\},\{2,6\},\{4,8\},\{3\},\{7\}]$. Two values in each vector are set in the first step
\begin{equation}
([\{4,8\},\{1,5\},?,?,?],[\{1,5\},?,\{4,8\},?,?]).
\end{equation}
We complete randomly the first vectors
\begin{equation}
([\{4,8\},\{1,5\},\{3\},\{2,6\},\{7\}],[\{1,5\},?,\{4,8\},?,?]).
\end{equation}
The remaining values of the right vector are deduced using step 3.
\begin{equation}
([\{4,8\},\{1,5\},\{3\},\{2,6\},\{7\}],[\{1,5\},\{3\},\{4,8\},\{7\},\{2,6\}]).
\end{equation}
Hence, we check that
	\begin{equation}
		\mathfrak s\left(	\begin{array}{|c|c|c|c|c|}\hline
			&&4\ 8&&\\\hline
			1\ 5&&&&\\\hline
			&3&&&\\\hline
			&&&&2\ 6\\\hline
			&&&7&\\\hline
		\end{array}
	\right)=\begin{array}{|c|c|c|c|c|}\hline
			&&\times&&\\\hline
			\times&&&&\\\hline
			&\times&&&\\\hline
			&&&&\times\\\hline
			&&&\times&\\\hline
		\end{array},
	\end{equation}
	as expected.
	
Observe that the inequality $2^{k-1}\leq n<2^{k}$ when $\sigma(0)= 0$ and $2^{k-1}< n\leq 2^{k}$ when $\sigma(0)\neq 0$. As a consequence, if $\sigma\neq\emptyset$, the state $E_{\sigma}$ is accessible from $\{(0,0)\}$ by a path of length $\lceil \log_{2}(n)\rceil$. But when  $\sigma(0)= 0$, we need a path of length $\lceil \log_{2}(n+1)\rceil$. This makes  a difference only when $n$ is a power of $2$. For instance, for $n=2$ we can not access the state $\{(0,0),(1,1)\}$ in less than $2$ steps while we need only one step to access to $\{(0,1),(1,0)\}$.
\end{example}
\begin{example}\label{fullTab}
	Let $k$ such that $2^{k-1}<m\leq 2^{k}$.
	We set 
	\begin{equation}\mathtt PP_{n}=[\{1,\dots,2^{k}\},\{2^{k}+1,\dots,2^{k+1}\},\dots,\{2^{k+n-2}+1,\dots,2^{k+n-1}\}], \end{equation}
		and $\mathtt P\Lambda_{m}=[p\lambda_{0},\dots,p\lambda_{m-1}]\in\mathcal L^{(k+n-1)}_{m}$ with
		\begin{equation}p\lambda_{i}=\{m-i+2^{k}\alpha\mid\alpha\in\IntEnt{2^{n}}\}\cup \{2^{k}(\alpha+1)-i\mid\alpha\in\IntEnt{2^{n}}\},\end{equation}
		for $i\in \IntEnt{2^{k}-m}$, and
		\begin{equation}p\lambda_{j}=\{m-j+2^{k}\alpha\mid\alpha\in\IntEnt{2^{n}}\},\end{equation} for $j\in\{2^{k}-m,\dots,m-1\}$.
	We let the reader check that $\mathtt P\Lambda_{m} \in\mathcal L^{(k+n-1)}_{m}$  and $\mathtt PP_{n}\in\mathcal R^{(k+n-1)}_{m}$. We have
	\begin{equation}
	\mathfrak s\left([\mathtt P\Lambda_{m},\mathtt PP_{n}]\right)=\IntEnt m\times\IntEnt n,
	\end{equation}
	because $m-i+2^{k}(2^{j}-1)\in p\lambda_{i}\cap p\rho_{j}$ for any $(i,j)\in \IntEnt m\times\IntEnt n$.
	For instance,
	\begin{equation}
		[\mathtt P\Lambda_{3},\mathtt PP_{5}]\sim\begin{array}{|c|c|c|c|c|}
		\hline{}
		\underline 3\ 4&\underline7\ 8&11\ 12\ \underline{15}\ 16&19\ 20\ 23\ 24\ 27\ 28\ \underline{31}\ 32&35\ 36\ 39\ 40\ 43\ 44\ 47\ 48\\
		&&&& 51\ 52\ 55\ 56\ 59\ 60\ \underline{63}\ 64\\\hline
		\underline{2} &\underline 6&10\ \underline{14}&18\ 22\ 26\ \underline{30}&34\  38\ 42\ 46\ 50\ 54\ 58\ \underline{62}\\\hline
		\underline 1&\underline 5&9\ \underline{13}& 17\ 21\ 25\ \underline{29}& 33\ 37\ 41\ 45\ 49\ 53\ 57\ \underline{61}\\\hline
		\end{array}.
	\end{equation}
	The underlined numbers correspond to the elements  $m-i+2^{k}(2^{j}-1)=2^{j+2}-(i+1)$, for $i=0\dots 2$ and $j=0\dots 4$.
\end{example}
This last example allows us to prove that many other states are reachable in $\mathfrak{Shuf}(M^{m}_{F_{1}},M^{n}_{F_{2}})$. Let us illustrate this point. Let $E$ be a state such that there exist $i_{1},i_{2}\in\IntEnt m$, $j_{1},j_{2}\in\IntEnt n$ satisfying $\{(i_{1},j)\in E\mid j\in\IntEnt n\}\subset \{(i_{2},j)\in E\mid j\in\IntEnt n\}$ and  $\{(i,j_{1})\in E\mid i\in\IntEnt m\}\subset \{(i,j_{2})\in E\mid i\in\IntEnt m\}$. If we denote by $\left( b\atop a\right)$ the map sending $a$ to $b$ and letting unchanged the other numbers, we have
\begin{eqnarray}
	E\cdot\left(\binom {i_{2}}{i_{1}},\binom{j_{2}}{j_{1}}\right)&=&{}
	\{(i,j)\in E\mid i\neq i_{1}\}
	\cup \{(i_{2},j)\mid (i_{1},j)\in E\}\nonumber{}
	\\&&\cup{}
	\{(i,j)\in E\mid j\neq j_{1}\}
	\cup \{(i,j_{2})\mid (i,j_{1})\in E\}\\
	&=&E\setminus\{(i_{1},i_{2})\}.\nonumber
\end{eqnarray}
Applying successively many times this property from $\IntEnt m\times\IntEnt n$ 
(which is reachable from Example \ref{fullTab}), we find that any state $E$ such that $\IntEnt m\times\{i\} \cup \{j\}\times\IntEnt n\subset E$, for some  $(i,j)\in \IntEnt m\times\IntEnt n$ is also reachable.
\begin{example}
	\[{}
	\begin{array}{|c|c|c|}
		\hline \times&\times&\times\\
		\hline \times&\times&\times\\
		\hline \times&\times&\times\\\hline
	\end{array}\mathop{\longrightarrow}^{\left(\binom 01,\binom01\right)}
\begin{array}{|c|c|c|}
		\hline \times&\times&\times\\
		\hline \color{white}\times&&\\
		\hline \times&\times&\times\\\hline
	\end{array}\cup	\begin{array}{|c|c|c|}
		\hline \times&\color{white}\times &\times\\
		\hline \times&\ &\times\\
		\hline \times&\ &\times\\\hline
	\end{array}=\begin{array}{|c|c|c|}
		\hline \times&\times &\times\\
		\hline \times&&\times\\
		\hline \times&\times&\times\\\hline
	\end{array}
	\mathop{\longrightarrow}^{\left(\binom 01,\binom02\right)}
	\begin{array}{|c|c|c|}
		\hline \times&\times &\times\\
		\hline \color{white}\times&&\\
		\hline \times&\times&\times\\\hline
	\end{array}\cup
	\begin{array}{|c|c|c|}
		\hline \times&\times &\color{white}\times\\
		\hline \times&&\color{white}\times\\
		\hline \times&\times&\color{white}\times\\\hline
	\end{array}
	=\begin{array}{|c|c|c|}
		\hline \times&\times &\times\\
		\hline \times&&\\
		\hline \times&\times&\times\\\hline
	\end{array}
	\mathop{\longrightarrow}^{\left(\binom 02,\binom02\right)}
	\begin{array}{|c|c|c|}
		\hline \times&\times &\times\\
		\hline \times&&\\
		\hline \times&\times&\\\hline
	\end{array}
	\]
\end{example}{}
As a consequence, by applying the inclusion-exclusion principle one obtains
\begin{equation}
	{\mathrm sc}_{\shuffle}(m,n)\geq \sum_{k=1}^{m}\sum_{\ell=1}^{n}(-1)^{k+\ell}\binom mk\binom n\ell2^{(m-k)(n-\ell)}>2^{(m-1)(n-1)}.
\end{equation}
Let us call \emph{dense} the states $E$ such that  $\{(i_{1},j)\in E\mid j\in\IntEnt n\}\subset \{(i_{2},j)\in E\mid j\in\IntEnt n\}$ implies $i_{1}=i_{2}$ and  $\{(i,j_{1})\in E\mid i\in\IntEnt m\}\subset \{(i,j_{2})\in E\mid i\in\IntEnt m\}$ implies $j_{1}=j_{2}$. We denote by $\mathcal{D}_{m,n}$ the set of dense $(m,n)$-states.\\
Conjecture \ref{conj1} reduces to the following one
 \begin{conj}\label{conj2}
	 \begin{equation}
		\mathcal D_{m,n}\subset\mathfrak s\left(\mathcal U_{m,n}\right).
	 \end{equation}
 \end{conj} 
 \begin{example} There are two kinds of dense states in $\mathcal D_{3,3}$ : those that have exactly one cross  by line and  column and those that have exactly two crosses by line and  by column. The states of first kind are clearly in $ s\left(\mathcal U_{3,3}\right)$ by using Example \ref{permutations}. There are $6$ states of second kinds that can be obtained the ones from the others by permuting the lines or the columns. They also belong  in  $s\left(\mathcal U_{3,3}\right)$ since we have
 \begin{equation}\begin{array}{ccc}
	  \begin{array}{|c|c|c|}\hline \times&\times&\\\hline&\times&\times\\\hline\times&&\times\\\hline \end{array}=
	  \mathfrak s\left(\ \begin{array}{|c|c|c|}\hline5\ 8&3&\\\hline&7&2\\\hline1\ 4&&6\\\hline \end{array}\ \right)
	  &\ &
	    \begin{array}{|c|c|c|}\hline \times&\times&\\\hline\times&&\times\\\hline&\times&\times\\\hline \end{array}=
	    \mathfrak s\left(\ \begin{array}{|c|c|c|}\hline5\ 8&3&\\\hline1\ 4&&6\\\hline&7&2\\\hline \end{array}\ \right)
	  \\ \\
	   \begin{array}{|c|c|c|}\hline \times&&\times\\\hline\times&\times&\\\hline&\times&\times\\\hline \end{array}=
	   \mathfrak s\left(\ \begin{array}{|c|c|c|}\hline5\ 8&&3\\\hline1\ 4&6&\\\hline&2&7\\\hline \end{array}\ \right)
	  {}
	  &\ &
	    \begin{array}{|c|c|c|}\hline \times&&\times\\\hline&\times&\times\\\hline\times&\times&\\\hline \end{array}=
	    \mathfrak s\left(\ \begin{array}{|c|c|c|}\hline5\ 8&&3\\\hline&2&7\\\hline1\ 4&6&\\\hline \end{array}\ \right)\\
	    \\
	   \begin{array}{|c|c|c|}\hline &\times&\times\\\hline\times&&\times\\\hline\times&\times&\\\hline \end{array}=
	  \mathfrak s\left(\ \begin{array}{|c|c|c|}\hline&6\ 8&3\\\hline1&&7\\\hline5&2\ 4&\\\hline \end{array}\ \right){}
	  &\ &
	    \begin{array}{|c|c|c|}\hline &\times&\times\\\hline\times&\times&\\\hline\times&&\times\\\hline \end{array}=
	    \mathfrak s\left(\ \begin{array}{|c|c|c|}\hline&3&6\ 8\\\hline1&7&\\\hline5&&2\ 4\\\hline \end{array}\ \right).
	    
	  \end{array}
 \end{equation}
 \end{example}
 \section{Conclusion}
 We give a new approach to study the state complexity of the shuffle operation. Unfortunately, we did not succeed in proving  the bound given by Brzozowski \textit{et al.}, but as we have translated the problem into a combinatorial one, independent from its language theoretical definition, we hope we have open up new research opportunities for this problem.    
\bibliography{../COMMONTOOLS/biblio,../COMMONTOOLS/bibjg}
\end{document}